\documentclass[letter, 12pt]{article}%
\usepackage{setspace}
\usepackage[top=1.25in, bottom=1.25in, left=1.25in, right=1.25in]{geometry}

\usepackage[title]{appendix}
\usepackage[demo]{graphicx}
\usepackage{tikz}
\usepackage{subcaption}
\usepackage{float}
\usepackage{fancyvrb}
\usepackage{boxedminipage}
\usepackage{graphicx} 
\usepackage{soul}
\usepackage{pgfplots}

\usepackage{wrapfig}
\usepackage{enumerate}
\usepackage{enumitem}
\usepackage{bibentry,natbib}

\usepackage{fancyref}
\usepackage[colorlinks, linkcolor=red, anchorcolor=blue,
citecolor=blue]{hyperref}

\usepackage{footmisc}
\usepackage{mathrsfs,amssymb}
\usepackage{amsmath}
\usepackage{color}
\usepackage{multirow}
\usepackage{multicol}
\usepackage{verbatim}
\usepackage{amsthm}
\usepackage{amsfonts}
\usepackage{amssymb}%
\usetikzlibrary{decorations.markings}

\usepackage{accents}
\newlength{\dhatheight}

\newcommand{\ubar}[1]{\underaccent{\bar}{#1}}

\newcommand{\cI}{\mathcal{I}}
\newcommand{\bt}{\mathbf{t}}

\providecommand{\U}[1]{\protect\rule{.1in}{.1in}}
\usetikzlibrary{shapes.geometric, arrows}
\usetikzlibrary{matrix,arrows,calc,shapes,decorations.pathreplacing}

\newtheorem{thm}{Theorem}

\newtheorem{example}{Example}
\newtheorem{lemma}{Lemma}
\newtheorem{prop}{ Proposition}

\newtheorem{definition}{ Definition}

\newtheorem{cl}{Claim}

\newtheorem*{rem*}{Remark}
\newtheorem*{prop*}{ Proposition}

\makeatletter

\renewenvironment{abstract}
{\small
	\begin{center}
		\bfseries \abstractname\vspace{-.1em}\vspace{0pt}
	\end{center}
	\list{}{
		\setlength{\leftmargin}{1cm}%
		\setlength{\rightmargin}{\leftmargin}%
	}%
	\item\relax}
{\endlist}

\let\OLDthebibliography\thebibliography
\renewcommand\thebibliography[1]{
	\OLDthebibliography{#1}
	\setlength{\parskip}{5pt}
	\setlength{\itemsep}{2pt plus 0.4ex}
}

\usepackage{titlesec}
\titlespacing*{\section}
{0pt}{.2ex plus .2ex minus .2ex}{.2ex plus .2ex}
\titlespacing*{\subsection}
{0pt}{.2ex plus .2ex minus .2ex}{.2ex plus .2ex}

\title{
	\Large 
	Information Design in Allocation with Costly Verification\thanks{We would like to thank Teddy Mekonnen, Satoru Takahashi, Rakesh Vohra, organizers and anonymous referees from WINE 2022 conference, and participants of various seminars and conferences for helpful comments.}
\vspace{-3mm}}

\author{\vspace{-3mm} Yi-Chun Chen\thanks{Department of Economics and Risk Management Institute, National University of Singapore, Singapore. ecsycc@nus.edu.sg.} 
	\and 
	Gaoji Hu\thanks{School of Economics, Shanghai University of Finance and Economics, and the Key Laboratory of Mathematical Economics (SUFE), Ministry of Education, Shanghai, China. hugaoji@sufe.edu.cn} 
	\and
	Xiangqian Yang\thanks{Department of Economics, Hunan University, Hunan, China. yangxiangqian@hnu.edu.cn}
} 

\date{ October 12, 2022 
}

\begin{document}
	\maketitle
	
	\vspace{-8mm}
	
	\begin{abstract}
		
		\begin{spacing}{1.35}
			
		A principal who values an object allocates it to one or more agents. Agents learn private information (signals) from an information designer about the allocation payoff to the principal. Monetary transfer is not available but the principal can costly verify agents' private signals. The information designer can influence the agents’ signal distributions, based upon which the principal maximizes the allocation surplus. An agent’s utility is simply the probability of obtaining the good.
		With a single agent, we characterize (i) the agent-optimal information, (ii) the principal-worst information, and (iii) the principal-optimal information. Even though the objectives of the principal and the agent are not directly comparable, we find that any agent-optimal information is principal-worst.
		Moreover, there exists a robust mechanism that achieves the principal's payoff under (ii), which is therefore an optimal robust mechanism. Many of our results extend to the multiple-agent case; if not, we provide counterexamples. (147 words)
		
	\end{spacing}

		
		
	\end{abstract}

	\vspace{-5mm}
	
	


	\renewcommand*{\thefootnote}{\arabic{footnote}}

	\parskip=1.5pt
	\setlength{\abovedisplayskip}{2.5pt}
	\setlength{\belowdisplayskip}{2.5pt}
	
	\setlist[itemize]{itemsep = -.5ex, topsep= -.5ex}
	\setlist[enumerate]{itemsep = -.5ex, topsep= -.5ex}
	
	\newpage
	\begin{spacing}{1.5}
		\tableofcontents
	\end{spacing}
	
	\newpage
	
	\section{Introduction}
	
	\label{sec:intro}
	
	%
	
	This paper studies information design in the context of allocation with costly verification \`{a} la \cite{ben2014optimal}. Particularly, a principal who values an object allocates it to one or more agents. Agents learn private information (signals) from an information designer about the allocation payoff to the principal. Monetary transfer is not available but the principal can costly verify agents' private signals. The information designer can influence the agents’ signal distributions, based upon which the principal maximizes the allocation surplus. An agent’s utility is simply the probability of obtaining the good.

	One can motivate agent-optimal information design and principal-optimal information design, together or separately, in different applications. Here are some examples from \cite{ben2014optimal} where we can fit in an information designer: 
	A venture capital firm may need to choose which of a set of competing startups to fund; its decision relies on a rating/auditing/marketing/consulting firm which provide information with different precision levels. A government may need to choose in which town to locate a new hospital; both the government and the people in town rely on information possessed by the town council. A funding agency may have a grant to allocate among several competing researchers; an external grant referee provides information that either favors the funding agency or the researchers.

	With a single agent, we characterize (i) the agent-optimal information, (ii) the principal-worst information, and (iii) the principal-optimal information. For concrete examples, making the signal distribution the least informative is principal-worst and the most informative being principal-optimal. An agent-optimal information pools information above a cutoff signal and fully reveals information below the cutoff. 
	
	By characterizing (i) and (ii), we find that any agent-optimal information is principal-worst, but not the converse. This result resembles the equivalence between agent-optimal and principal-worst information in \cite{roesler2017buyer}. However, both the problem formulations and the results are different. In \cite{roesler2017buyer}, the buyer and the seller are roughly playing a zero-sum game: the monetary surplus of trading is split between them and payoffs are perfectly transferable. Yet in our paper, the objectives of the principal and the agent are not directly comparable: the principal maximizes the allocation surplus whereas the agent maximizes the probability of obtaining the good. Actually, the two parties' interests are overlapped, which makes the opposition in solutions somewhat surprising at the first glance. In terms of results, \cite{roesler2017buyer} offers an equivalence as a consequence of the formulation, whereas our connection between the agent-optimal and principal-worst information is only one direction: some principal-worst information distributions are not agent-optimal. 
	
	The principal's payoff under a principal-worst information provides an upper bound for the payoff that can be achieved by a ``robust'' mechanism which does not depend on details of the agent's type distribution. We find a robust mechanism that does achieve such an upper bound payoff, which is therefore an optimal robust mechanism. The mechanism resembles the optimal favored-agent mechanism in \cite{ben2014optimal}, but there is a natural difference: the optimal favored-agent mechanism uses a threshold that is calculated from the agent's type distribution and the associated checking cost, whereas the optimal robust mechanism uses the mean value of the agent's type as threshold.
	
	With multiple agents, agent-optimal information maximizes the total probability of agents' obtaining the good. Compared with the prior distribution, under some agent-optimal information, all agents can be better off; while under some other agent-optimal information, some agents get worse off. Moreover, agent-optimal informations may deliver different payoffs to the principal, which implies that an agent-optimal information need \textit{not} be principal-worst.

	In contrast, the robust-mechanism result extends perfectly to the multiple-agent setting. An optimal robust mechanism simply allocates the good to the agent who has the highest expected type. This result is more general than it appears. Particularly, in similar models of allocation without transfer \citep{mylovanov2017optimal,chua2019optimal,li2021mechanism}, the aforementioned simple mechanism is an optimal robust mechanism. Moreover, allowing for the correlated signal distributions does not affect the result. 

	The rest of this section reviews the literature. In Section \ref{sec:model}, we set out the model, the allocation problem and the associated information design problems. In Section \ref{sec:benchmark}, we present a solution to the benchmark allocation problem, which we attribute to \cite{ben2014optimal}. Section \ref{sec:agent-optimal} characterizes the agent-optimal information. Section \ref{sec:principal} characterizes the principal-worst information and the principal-optimal information, and explores the implications of those characterizations. In Section \ref{sec:multiple}, we extend our model to allow for multiple agents. Section \ref{sec:conclusion} concludes. All omitted proofs are in Appendix \ref{sec:proofs}.

	\subsection*{Related literature}
	\label{sec:literature}

	Our paper is related to three streams of literature: allocation with costly verification, information design, and robust mechanism design.

	The literature on costly state verification is initiated by \cite{townsend1979optimal} which studies optimal debt contracts; see also \cite{gale1985incentive} and \cite{mookherjee1989optimal}. Unlike those earlier papers, \cite{ben2014optimal} studies the role of costly verification in allocation problems \textit{without} monetary transfer. Since then, their model has been modified or extended to different directions; see, e.g., \cite{mylovanov2017optimal}, \cite{halac2020commitment}, \cite{erlanson2020costly}, \cite{li2021mechanism}, \cite{epitropou2019optimal}, \cite{chua2019optimal} and  \cite{kattwinkel2019costless}, among many others.\footnote{\ \cite{ben2019mechanisms} studies a general model of mechanism design with evidence, which differs significantly from the allocation problem of \cite{ben2014optimal}. However, it turns out that the results there can be used to solve allocation problems in, say, \cite{erlanson2020costly} and \cite{chua2019optimal}.} We depart from \cite{ben2014optimal} by introducing an information designer who can influence the prevailing information in the allocation problem.

	We formulate the information design problems  \`{a} la \cite{roesler2017buyer} in allocation problems with costly verification. Although the analysis of information design advances rapidly in the auction literature (see, e.g., \cite{bergemann2017first}, \cite{yang2019buyer, yang2021efficient} and \cite{chen2020information}) and in other areas (see \cite{bergemann2019information} for a comprehensive survey), its counterpart in the aforementioned allocation setting is rarely studied. One exception is \cite{kattwinkel2019costless}, which shows in a single-agent model that a principal who privately observes a signal correlated with the agent's type does not profit from persuading the agent to reveal his information with any form of information design. Information design in their paper pertains to the disclosure of the principal's private signal to the agent who knows his own type. However, in our setting, there is an independent information designer who designs the entire uncertainty in the underlying allocation problem, perceived by both the principal and the agent.\footnote{\ Another difference is that they focus on the correlation between the principal's signal and the agent's type, whereas correlation is not our focus. Nevertheless, both their results and ours convey the (loosely) similar idea that making information more precise is better for the principal.}
	
	Our exercise of robust mechanism design follows the convention of, e.g., \cite{bergemann2016informationally}, \cite{du2018robust}, \cite{koccyiugit2020distributionally}, \cite{brooks2021optimal} and \cite{he2022correlation}. Focusing on the allocation problem of \cite{ben2014optimal}, our paper is most closely related to \cite{bayrak2017optimal} and \cite{bayrak2022optimal}. Instead of solving the principal-worst information design problem first, like we do, \cite{bayrak2017optimal} directly works on the mechanism design problem with uncertain type distributions. Particularly, they adopt the linear programming approach initiated by \cite{vohra2012optimization}, which necessitates the assumption of discrete (and finitely many) types that we do not impose.\footnote{\ Our paper also differs from \cite{bayrak2017optimal} in the following two aspects: First, they assume i.i.d. distributions of agent types whereas we allow for arbitrary asymmetry among agents; see Section \ref{sec:multiple}. Second, their robust mechanism is with respect to a set of distributions that is notably different from ours, i.e., they either assume first-order stochastic dominance within the set of possible distributions or restrict attention to a binary set, whereas we work with a general class of sets of distributions that have the same mean without imposing other structures.} Contemporary with our paper, \cite{bayrak2022optimal} studies a similar robust mechanism design problem with respect to two sets of distributions, where the second is more related to our analysis. The Markov ambiguity set, as they call it, is parametrized by the lower and upper bounds on the expected type for each agent. When the lower and upper bounds coincide for all agents, this special case of the general Markov ambiguity set coincides with a special case of our general set of possible distributions. Therefore, in terms of robust mechanism design, the two papers complement each other in both methods and results.

	\section{The model}
	\label{sec:model}

	\subsection{Preliminaries}
	\label{sec:preliminaries}
	
	We slightly modify the setup of \cite{ben2014optimal} first, to open up a window for information design. Namely, a principal is to decide whether or not to allocate an indivisible good to an agent; the case of multiple agents will be studied in Section \ref{sec:multiple}. 
	
	The net value to the principal of allocating the good is $t - R$, where $R$ is the principal's reserve value and $t$ is the agent's type. $R$ is fixed and is common knowledge among the agents and the principal. In contrast, there is an underlying distribution for $t$ that is over the interval $T:=[\ubar{t},\bar{t}]$, where $0 \leq  \ubar{t} <  R < \bar{t}  < \infty$. We denote by $F$ the distribution function of $t$ and $f$ the density; assume $f(t) > 0$ for all $t\in T$. Unlike in \cite{ben2014optimal}, we assume that neither the principal nor the agent knows the distribution $F$ or its realization $t$, but they may receive related information from an \textit{information designer}. 
	
	Following \cite{roesler2017buyer}, there is an information designer who can influence the ``private information'' that the agent receives and the ``distributional information'' that the principal perceives. More formally, the information designer knows the underlying distribution $F$ (but not its realization). She can design a statistical experiment which reveals some information about $t$ to the agent and the principal. Generally, this is through a joint distribution between $t$ and another random variable, which is called ``signal'' and denoted by $s$, subject to the constraint that the marginal distribution of $t$ is still $F$. The agent privately learns the realized signal $s$ and the principal only learns the distribution of $s$. 
	
	According to \cite{roesler2017buyer}, without loss of generality, we restrict attention to \textit{unbiased signals} such that $\mathbb{E}(t \vert s) = s$, where $s$ naturally belongs to $T$. Furthermore, the payoffs of both the principal and the agent are determined by the marginal distribution of the signal $s$ (since $t - R$ is linear in $t$ and signals are unbiased estimation of $t$). Thus, we further restrict our attention to the \textit{marginal distributions} of unbiased signals. Let $\mathcal{G}$ be the set of all cumulative distribution functions defined over $T$ and $\mathcal{G}_F \subseteq \mathcal{G}$ the set of marginal distributions of unbiased signals. Based upon the characterization of \cite{blackwell1953equivalent}, the set $\mathcal{G}_F$ is exactly the set of \textit{mean-preserving contractions (MPC)} of $F$, i.e. 
	$$
	\mathcal{G}_F = \left\lbrace G \in \mathcal{G}:  \int_{\ubar{t}}^{x} G(s) \mathrm{d} s \le \int_{\ubar{t}}^{x} F(t)\mathrm{d} t  \quad \forall x\in T\quad \text{ and } \quad \int_{\ubar{t}}^{\bar{t}} G(s) \mathrm{d} s = \int_{\ubar{t}}^{\bar{t}} F(t)\mathrm{d} t \right\rbrace.
	$$
	For example, $F\in \mathcal{G}_F$. Let $\mu = \mathbb{E}(t)$. Then the degenerate distribution $G = \delta(\mu)$ that assigns probability one to the atom $s = \mu$ is also in $\mathcal{G}_F$; we refer to $\delta(\mu)$ as the \textit{null information} distribution. In what follows, the information designer has the flexibility to choose a marginal distribution from $\mathcal{G}_F$.

	The principal can \emph{check} the agent's private signal $s$ at a cost $c > 0$. We interpret checking as obtaining information (e.g. by requesting documentation, interviewing the agent, or hiring outside evaluators) which perfectly reveals the signal of the agent. The cost to the agent of providing information is assumed to be zero. To avoid the trivial case, let us assume $R + c \le \max_{s\in T} s = \bar{t}$; otherwise the principal never checks the agent.
	
	We assume that the agent strictly prefers receiving the good to not receiving it. Consequently, we can take the agent's payoff to be the probability with which he receives the good. The intensity of the agents' preference plays no role in the analysis, so it is omitted. We also assume that the agent's reservation utility is less than or equal to his utility from not receiving the good. Since monetary transfers are not allowed, not receiving a good delivers the worst payoff to an agent. Consequently, individual rationality constraints do not bind and so are disregarded throughout.

	\subsection{Direct mechanisms}
	\label{sec:dm}
	
	%
	
	The mechanism design part is identical to that of \cite{ben2014optimal}, except that the type $t$ is replaced by signal $s$. 
	A general mechanism in this setting is a game that specifies (i) what messages the agent can report, (ii) which of those messages lead to checking (and with which probability), and (iii) depending on the reported message and checking outcome, whether the agent receives the good (and with which probability). We can apply the same argument as in the online appendix of \cite{ben2014optimal}, up to terminology replacement, to establish a revelation principle and to obtain the following two necessary conditions for any direct mechanism to be optimal: 
	\begin{enumerate}[leftmargin= 0.5\parindent]
		\item If the agent is checked and found lying, then the conditional probability of him receiving the good has to be zero.
		
		\item If the agent is checked and found truth-telling, then the conditional probability of him receiving the good has to be one.
	\end{enumerate}
	Hence, we restrict our attention to the ``(simplified) direct mechanisms.''
	
	Formally, a \emph{direct mechanism} consists of (i) a checking rule $q$ that maps each reported signal $s \in T$ to a checking probability $q(s)\in [0,1]$, which is also the checking-and-assigning probability since the agent is truthful on equilibrium path, and (ii) a total allocation rule $p$ that maps each reported signal $s \in T$ to a total assignment probability $p(s)\in [0,1]$. Naturally, $p(s) \ge q(s)$ for all $s\in T$ as the former also accommodates the assignment probability conditional on not being checked. The artificial game tree in Figure \ref{fig:direct} illustrates the simplified direct mechanism.
	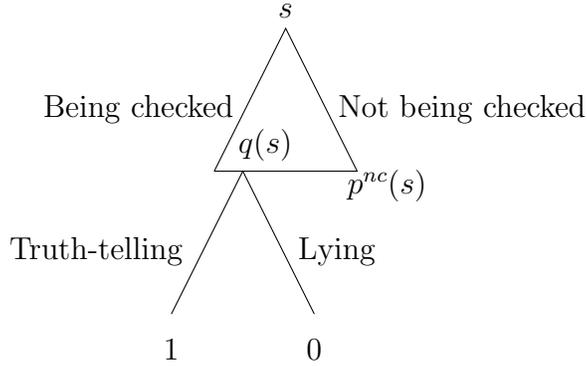
\begin{figure}[H]
		\begin{center}
			\begin{tikzpicture}[scale=0.95]
				\draw (0,0)node[above,xshift=-10mm,yshift=5mm]{Being checked}--(1,2) node [above]{$s$} -- (2,0)node[above,xshift=14mm,yshift=5mm]{Not being checked}--(0,0);
				\draw (-.6,-2) node[above,xshift=-10mm,yshift=5mm]{Truth-telling} --(0.4,0)node[above, xshift = 3mm]{$q(s)$}--(1.4,-2) node[above,xshift=3mm,yshift=5mm]{Lying};
				\node at (-.6,-2.5) {$1$};
				\node at (1.4,-2.5) {$0$};
				\node at (2,-0.2) [xshift=4mm] {$p^{nc}(s)$};
			\end{tikzpicture}
		\end{center}
	\vspace{-3mm}
		\caption{The simplified direct mechanism, where $p^{nc}(s)$ is the assignment probability conditional on no checking and $p(s) = q(s) \cdot 1 + (1-q(s)) \cdot p^{nc}(s)$.}
		\label{fig:direct}
	\end{figure}

	\subsection{Mechanism design and information design}
	\label{sec:md}
	
	%
	%
	Given a signal distribution (interchangeably, designed information) $G\in \mathcal{G}_F$, the principal selects a mechanism $(p,q)$ to maximize her expected net payoff:
	\begin{align}
		\max_{p,q} \quad & \mathbb{E}_{G}  \left[ p(s) (t-R) -q(s)c \right] \label{md} \tag{MD} \\
		\text{ subject to } \quad
		& p(s) \in [0,1], \quad \forall s\in T,  \label{eq:pfqf} \\
		& q(s) \in [0,1], \quad \forall s\in T, \\
		&q(s)\leq p(s) , \quad \forall s\in T, \label{eq:n}\\
		&p(s) \ge p(s') - q(s'), \quad \forall s, s'\in T.\label{eq:ic}
	\end{align}
	(\ref{md}) is the objective of the principal; sometimes we also refer to the entire constrained maximization as problem (\ref{md}). Constraints (\ref{eq:pfqf})-(\ref{eq:n}) are the feasibility constraints. Constraint (\ref{eq:ic}) is the incentive compatibility constraint to ensure that the agent always prefers truth-telling to lying. 
	
	Problem (\ref{md}) here is a simplification of the allocation problem in \cite{ben2014optimal}: Instead of having multiple agents with private information, we only have one agent with private information and another ``naive agent''\textemdash the principal\textemdash who has a constant type $R$. Therefore, the general results in \cite{ben2014optimal} fully characterize the solution to problem (\ref{md}), which is the benchmark of our analysis and we will review it shortly in Section \ref{sec:benchmark}. 
	
	The main focus of our paper is the information design problems associated with problem (\ref{md}) and its multiple-agent counterpart.
	We will study information design from both the agent's perspective and the principal's perspective, and set up optimization problems in corresponding sections.

	\section{Benchmark optimal mechanism for a given $G$}
	
	\label{sec:benchmark}

	This section presents the solution to the benchmark mechanism design problem (\ref{md}), which we adopt from \cite{ben2014optimal}. 
	
	The optimal mechanism is essentially unique and is determined by a single threshold $s^*$.\footnote{\ The word ``essential'' means that the optimal mechanism is unique up to manipulation on a set of measure zero. For example, if a mechanism is optimal, then the alternative mechanism that retains the good more frequently on a measure zero set is also optimal. See \cite{ben2014optimal} for a formal notion of equivalence between different mechanisms.} Particularly, $s^*$ is defined implicitly by the following equation
	\begin{equation}
		\label{eq:t}
		\mathbb{E}_G (s) = \mathbb{E}_G (\max \{s, s^*\}) - c.   \tag{TH}
	\end{equation}
	Since the equation is equivalent to $\int_{\ubar{t}}^{s^*} G(s) ds = c$, it has a unique solution unless $c = 0$. If $c = 0$, any $s^* \le \min supp(G)$ solves the equation, but it is natural to take $s^* =  \min supp(G)$ since it is the limit value of $s^*$ as $c$ goes to $0$ from above. Overall, there is a unique solution to (\ref{eq:t}) in $[ \min supp(G), \infty)$. We usually suppress the dependence of $s^*$ on $G$ when there is no ambiguity. For notational convenience, we denote by $t^*$ the threshold defined by equation (\ref{eq:t}) but using the underlying distribution $F$ instead of $G$. 
	
	If we view the principal as another ``agent $0$,'' then the principal has a constant type $R$ and we have a similar equation to define a threshold $t^*_0$ for the principal:
	$$
	\mathbb{E}_{G_0} (s_0) = \mathbb{E}_{G_0} (\max \{s_0, s^*_0\}) - c_0,
	$$
	where $s_0 \equiv R$ and $c_0 = 0$. Obviously, any $s_0^* \le R = \min supp(G_0)$ solves the equation.

	\begin{prop}[\cite{ben2014optimal}]
		\label{prop:benchmark}
		
		The following mechanism is essentially the unique optimal mechanism:
		\begin{enumerate}[leftmargin = 0.5\parindent]
			\item If $s^* - c \ge R$, then the agent receives the good without being checked.
			
			\item If $s^* - c < R$, then proceed as follows:
			\begin{enumerate}
				\item If $s - c < R$, then the principal retains the good.
				
				\item If $s - c \ge R$, then the agent is checked (truth-telling on equilibrium path) and allocated the good. 
			\end{enumerate}
		\end{enumerate}
	\end{prop}
	
	For completeness and easy reference, a proof of Proposition \ref{prop:benchmark} is included in Appendix \ref{sec:proofs}; readers who are familiar with \cite{ben2014optimal} would immediately see that the mechanism in Proposition \ref{prop:benchmark} is simply a special case of their optimal favored-agent mechanism for multiple agents.
	
	The optimal mechanism delivers the following payoff to the principal:
	\begin{equation}
		\label{eq:y}
		Y:= \begin{cases}
			\int_{\ubar{t}}^{\bar{t}}(s - R) \mathrm{d} G(s) & \text{ if } s^* - c \ge R\\
			\int_{R + c }^{\bar{t}} (s - c - R) \mathrm{d} G(s) & \text{ if } s^* - c < R,
		\end{cases} 
		\tag{Y}
	\end{equation}
	where the former can also be written as  $\mathbb{E}_{G}(s) - R = \mathbb{E}_{F}(t) - R = \mu - R$.

	We break ties in favor of the agent to facilitate the study of the agent-optimal information design problem in Section \ref{sec:agent-optimal}. Specifically, there are two kinds of ties to consider. First, when $s - c = R$, the principal is indifferent between retaining the good and allocating the good to the agent with checking. Second, from the proof of the proposition in Appendix \ref{sec:proofs}, we know that when $s^* - c  =  (resp. <) \ R$, 
	\begin{align}
		\mu - R  \ = \ &  \int_{\ubar{t}}^{\bar{t}}(s - R) \mathrm{d} G(s) \nonumber \\
		\ = \ &    (resp. <) \ \int_{R + c }^{\bar{t}} (s - c - R) \mathrm{d} G(s), \label{eq:payoffcomparison} \tag{C}
	\end{align}
	i.e. the principal is indifferent between allocating the good without checking and allocating the good only if $s - c \ge R$ and the agent tells the truth. In both cases, we opt to let the agent receive the good. Of course, the principal is always indifferent between tie-breaking rules.\footnote{\ One may break ties the other way around if interested in the agent-worst information design problem, which, however, is not the focus of the current paper.}

	\section{Agent-optimal information design}
	
	\label{sec:agent-optimal}

	We characterize the agent-optimal information in this section. The agent wants to obtain the good with the largest probability, given that the principal employs an optimal mechanism to maximize her expected payoff. 
	
	Before we set up the agent-optimal information design problem, we first observe that mean preserving contraction leads to a larger threshold. Recall that $s^*$ is uniquely defined on $[\min supp(G), \infty)$. Moreover, if $c > 0$, then $s^* > \ubar{t}$.
	
	\begin{lemma}
		\label{lem:sbiggerthant}
		$s^* \ge t^*$ for any $G\in \mathcal{G}_F$.
	\end{lemma}

	Since the agent obtains the good with probability one at all types if $s^* - c \ge R$, and with probability one only when $s - c \ge R$ if $s^* - c < R$ (Proposition \ref{prop:benchmark}), the \textit{agent-optimal information design} problem is as follows:

	\begin{enumerate}[leftmargin = 0.5\parindent]
		\item If there exists $G\in \mathcal{G}_F$ such that $s^* - c \ge R$, then any such $G$ maximizes the agent's probability of obtaining the good.
		
		\item If for all $G \in \mathcal{G}_F$ we have $s^* - c < R$, then we need to solve 
		\begin{equation}
			\label{eq:idao}
			\max_{G\in \mathcal{G}_F} \quad 1 - G^{-}(R + c), \tag{ID-AO}
		\end{equation}
		where $G^{-}(R + c) := \lim_{s \uparrow R + c} G(s)$. (\ref{eq:idao}) becomes trivial if $R + c > \bar{t}$, since $G^{-}(R + c) \ge G(\bar{t}) = 1$ implies $G^{-}(R + c) = 1$ for any $G \in \mathcal{G}_F$. This is why we have assumed $R + c \le \bar{t}$ in Section \ref{sec:preliminaries}.
		
	\end{enumerate}

	\begin{prop}
		\label{prop:agent-optimal}
		Given $R$, $F$ and $c$:
		\begin{enumerate}[leftmargin = 0.5\parindent]
			\item If $\mu\geq R$, then there exists a distribution $G\in\mathcal{G}_F$ such that $s^*-c\geq R$, and any such $G$ maximizes the agent's payoff. In particular, null information satisfies 
			$
			s^* - c = \mathbb{E}_G(s) = \mu \ge R
			$
			and is agent-optimal.
			\item If $\mu<R$, then an information $G\in \mathcal{G}_{F}$ is agent-optimal if and only if $G$ has an atom $R + c$ with probability $1- F(s^{\dagger})$, where $s^{\dagger}$ is determined by 
			\begin{equation}
				\label{eq:s1}
				R + c = \mathbb{E}_{F}\left(t \big \vert t\in \left[s^{\dagger}, \bar{t}\right] \right).\footnote{\ The threshold $s^{\dagger}$ uniquely exists and satisfies $s^{\dagger} < R + c \le \bar{t}$ since $f(t)>0$ for all $t\in T$.}
			\end{equation}
			In particular, the following information is agent-optimal:
			\begin{equation}
				\label{eq:aoi1}
				\hat{G}(s) =\begin{cases}
					F(s) & \text{ for all } s\in[\underline{t},s^{\dagger})\\
					F(s^{\dagger}) & \text{ for all }s \in[s^{\dagger},R+c)\\
					1&  \text{ for all }s \in[R+c,\bar{t}].
				\end{cases} \tag{AOI}
			\end{equation}
		\end{enumerate}
	\end{prop}

	The proposition provides a full characterization of the agent-optimal information. The particular distribution $\hat{G}$ in part 2 is derived from $F$ by concentrating the probability mass of $F$ over $[s^{\dagger}, \bar{t}]$ onto a single point $R+c$, as required by (\ref{eq:s1}), and maintaining $F$ on $[\underline{t},s^{\dagger})$. This is illustrated in the left panel of Figure \ref{fig:agentoptimal}. Intuitively, (\ref{eq:idao}) says that reducing the probability mass on $[\ubar{t}, R + c)$ and increasing the probability mass on $[R + c, \bar{t}]$ would benefit the agent. This can be done to the largest extent, while maintaining the MPC constraint $G\in \mathcal{G}_F$, if we concentrate as much probability mass as possible to the single point $R + c$. The endogenously determined $s^{\dagger}$ captures such a limit. 
	
	However, $\hat{G}$ is not the unique agent-optimal information. For example, $\tilde{G}$ in the right panel of Figure \ref{fig:agentoptimal} is also agent-optimal. Generally, varying the lower part of an agent-optimal information does not affect its optimality as long as the MPC constraint is maintained.
	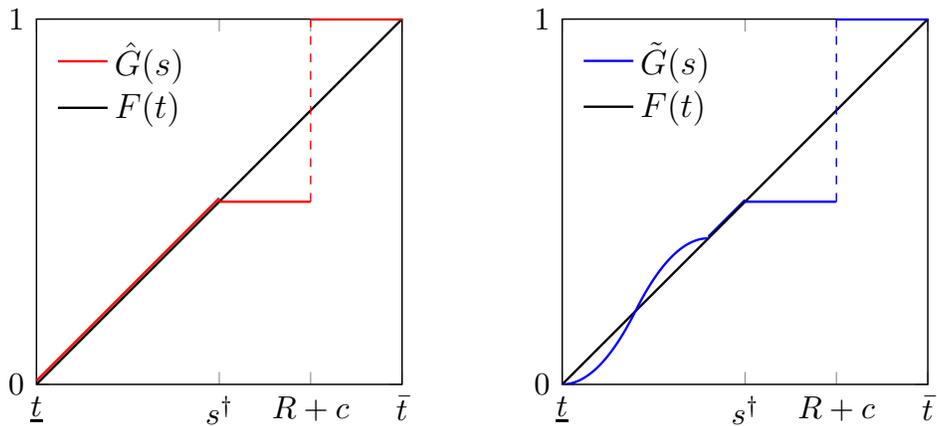
\begin{figure}[H]
		\begin{center}
			\begin{tikzpicture}[scale=1.1]
				\begin{axis}[
					tick label style={font=\small},
					xlabel={ },
					ylabel={ },
					ytick={ 0, 1},
					yticklabels={  0, 1},
					xtick={ 0,0.5,0.75,1},
					xticklabels={$\underline{t}$ ,$s^{\dagger}$, $R+c$, $\bar{t}$},
					no markers,
					line width=0.5pt,
					cycle list={{red,solid}},
					samples=200,
					smooth,
					domain=0:1.3,
					xmin=0, xmax=1,
					ymin=0, ymax=1,
					width=6cm, height=6cm,
					legend cell align=left,
					legend pos=  north west,
					legend style={draw=none,fill=none,name=legend},
					]
					\addplot[red,thick,domain=0:0.5]{x+0.01};
					\addplot[black,thick,domain=0:1]{x};
					\legend{$\hat{G}(s)$,$F(t)$};
					\addplot[red,thick,domain=0.5:0.75]{0.5};
					\addplot[red,thick,domain=0.75:1]{0.999};
					
					\addplot[red,dashed] coordinates {
						(0.75,0.5)
						(0.75,1)
					};
				\end{axis}
			\end{tikzpicture}
			\quad \quad \quad 
			\begin{tikzpicture}[scale=1.1]
				\begin{axis}[
					tick label style={font=\small},
					xlabel={ },
					ylabel={ },
					ytick={ 0, 1},
					yticklabels={ 0, 1},
					xtick={ 0,0.5,0.75,1},
					xticklabels={$\underline{t}$ ,$s^{\dagger}$, $R+c$, $\bar{t}$},
					no markers,
					line width=0.5pt,
					cycle list={{red,solid}},
					samples=200,
					smooth,
					domain=0:1.3,
					xmin=0, xmax=1,
					ymin=0, ymax=1,
					width=6cm, height=6cm,
					legend cell align=left,
					legend pos=  north west,
					legend style={draw=none,fill=none,name=legend},
					]
					\addplot[blue,thick,domain=0.4:0.5]{x+0.005};
					\addplot[black,thick,domain=0:1]{x};
					\legend{$\tilde{G} (s)$,$F(t)$};
					\addplot[blue,thick,domain=0.5:0.75]{0.5};
					\addplot[blue,thick,domain=0.75:1]{0.999};
					
					\addplot[blue,dashed] coordinates {
						(0.75,0.5)
						(0.75,1)
					};
					\addplot[blue,thick,domain=0:0.2]{5*x*x};
					\addplot[blue,thick,domain=0.2:0.4]{ 0.4-5*(x-0.4)*(x-0.4)};
				\end{axis}
			\end{tikzpicture}
		\end{center}
		\caption{Two examples of agent-optimal information.}
		\label{fig:agentoptimal}
	\end{figure}

Technically, some features of an agent-optimal information are implied by the results of \cite{dworczak2019simple} and \cite{kleiner2021extreme}. In particular, given an information distribution, the equilibrium outcomes of both the principal and the agent under the optimal mechanism in Proposition \ref{prop:benchmark} can be reformulated as a special case of the principal's expected payoff and one of the agent's expected utility \`{a} la \cite{dworczak2019simple} and \cite{kleiner2021extreme}. Thus, applying Theorem 3 of the former or Proposition 2 of the latter, we know that some agent-optimal information partitions signals into intervals and is pooling or fully revealing for different intervals. Unlike their characterizations without explicit solutions, we fully characterize the solutions in a concrete context. Moreover, our direct and elementary proof facilitate economic interpretations and the study of other research questions, such as the one to be answered in Proposition \ref{prop:agentprincipal}.

	Now we turn to the principal's payoff. Given any information $G \in \mathcal{G}_{F}$, the principal's payoff from an optimal mechanism is (\ref{eq:y}).
	Since the two cases in Proposition \ref{prop:agent-optimal} are comparing $\mu$ and $R$, whereas the two cases in (\ref{eq:y}) are comparing $s^* - c$ and $R$, we shall clarify the connection to facilitate the calculation of principal payoff. The following lemma does so and it is also used in the proof of Proposition \ref{prop:agent-optimal}.
	
	\begin{lemma}
		\label{lem:etsmall}
		If $\mu < R$, then for any $G\in\mathcal{G}_F$, we have $s^* - c < R$.
	\end{lemma}

	Using the particular information $\hat{G}$ such that $\hat{G}(s) = F(s) $ for all $s\in[\underline{t},s^{\dagger})$, we can easily calculate the principal's payoff under an agent-optimal information:
	\begin{equation}
		\label{eq:yao}
		Y^{AO} :=
		\begin{cases}
			\mu - R & \text{ if } \mu \ge R\\
			0 & \text{ if } \mu < R.
		\end{cases}\tag{Y-AO}
	\end{equation}
	As we will see shortly, all agent-optimal information deliver the same payoff (\ref{eq:yao}) to the principal.

	We close this section by noticing a welfare consequence of information design. As documented by \cite{roesler2017buyer}, the buyer-optimal information structure can generate efficient trade which is generally not occurred without information design. Here in our context of allocation with costly verification, we have a similar observation that information design can make allocation more efficient. For example, when $\mu < R$ (which implies $t^* - c < R$), the optimal mechanism under the prior distribution $F$ allocates the good if and only if $t - c \ge R$ (Proposition \ref{prop:benchmark}). Note that it is ex post more efficient to allocate the good whenever $t \ge R$. Therefore, slightly pooling information around $R + c$ via mean-preserving contraction would improve allocation efficiency.

	\section{Principal-related information design}
	\label{sec:principal}

	\subsection{Principal-worst information design}

	\label{sec:principal-worst}

	The principal-worst information design problem is as follows: 
	\begin{align}
		\min_{G \in \mathcal{G}_F}\quad \max_{p,q} \quad & \mathbb{E}_{G}  \left[ p(s) (s-R) -q(s)c \right]      \label{eq:pw} \tag{PW}\\
		\text{ subject to } & (\ref{eq:pfqf})-(\ref{eq:ic})\nonumber
	\end{align}
	and its solution is characterized below.
	
	\begin{prop}
		\label{prop:principal-worst}
		Given $R$, $F$ and $c$:
		\begin{enumerate}[leftmargin = 0.5\parindent]
			\item If $\mu \ge R$, an information $G\in \mathcal{G}_{F}$ is principal-worst if and only if $s^* - c \ge R $. In particular, null information $\delta(\mu)$ is principal-worst.
			
			\item If $\mu<R$, an information $G\in \mathcal{G}_{F}$ is principal-worst if and only if $G(R+c) = 1$. In particular, null information $\delta(\mu)$ is principal-worst.
		\end{enumerate}
	\end{prop}

	Using null information, together with (\ref{eq:y}), we can easily calculate the principal's worst-case payoff:
	\begin{equation}
		\label{eq:ypw}
		Y^{PW} :=
		\begin{cases}
			\mu - R & \text{ if } \mu \ge R \\
			0 & \text{ if } \mu < R.
		\end{cases}\tag{Y-PW}
	\end{equation}
	\noindent Comparing (\ref{eq:yao}) and (\ref{eq:ypw}), we immediately see that the principal's payoff is the same under the agent-optimal information and the principal-worst information, i.e. $Y^{PW} = Y^{AO}$. This is surprising in the sense that the objectives of the two information design problems are very different: The agent-optimal problem concerns the \textit{probability} that the agent receives the good, whereas the principal-worst problem concerns the \textit{payoff} to the principal. 
	
	Here is the reason for such coincidence: (a) When $\mu \ge R$, the agent-optimal information leads to probability-one allocation without checking, which obviously achieves the principal's lower-bound payoff $\mu - R$. (b) When $\mu < R$, since $s^* - c < R $ for all $G\in \mathcal{G}_{F}$ (Lemma \ref{lem:etsmall}), allocating the good to the agent necessarily incurs checking for any $G\in \mathcal{G}_{F}$ and the principal's payoff is $\int_{R + c }^{\bar{t}} (s - c - R) \mathrm{d} G(s)$. We can rewrite the payoff as follows:
	$$
	\int_{R+c}^{\bar{t}} s \mathrm{d} G(s) - [1- G^{-}(R + c)] (R + c).
	$$
	So intuitively, maximizing $[1- G^{-}(R + c)]$ (agent-optimal) at least partially coincides with minimizing the principal's payoff. 
	
	Formally, we have the following proposition:
	\begin{prop}
		\label{prop:agentprincipal}
		Any agent-optimal information is principal-worst.
	\end{prop}
	Two remarks are in order. First, the converse of Proposition \ref{prop:agentprincipal} is not true. For example, null information is principal-worst, but it is usually not agent-optimal. Second, an agent-optimal distribution, say $\hat{G}$ in (\ref{eq:aoi1}), may be more informative than the null distribution in the sense that null distribution comes from coarsening of the underlying probability space. However, the proposition says that both of them are principal-worst, i.e., the additional informativeness of $\hat{G}$ does not help the principal to generate higher payoff.

	Proposition \ref{prop:agentprincipal} resembles the equivalence between buyer-optimal and seller-worst information in \cite{roesler2017buyer}. However, both the problem formulations and the results are different. In \cite{roesler2017buyer}, the buyer and the seller are playing a zero-sum game: the monetary surplus of trading is split between them and payoffs are perfectly transferable. Yet in our model, the agent's and the principal's interests are largely overlapped, which makes the opposition in solutions surprising at the first glance. In terms of results, \cite{roesler2017buyer} offers an equivalence as a consequence of the formulation, whereas our result is only one directional: some principal-worst information distributions are not agent-optimal. 
	
	\subsection{Robust mechanism design}
	\label{sec:robust}
	
	Characterizing principal-worst information facilitates the study of robust mechanism design. More precisely, suppose for now that there is \textit{no} information designer. Suppose the principal faces ambiguity of the agent's type distribution; she only knows that the distribution lies in $\mathcal{G}$ and its expectation is $\mu = \mathbb{E}(t)$. What is the principal's optimal \textit{robust mechanism} that depends only on $\mu$ but not on the specific type distribution? 
	This problem can be formulated as follows:
	\begin{align}
		\max_{p,q} \quad \min_{\tiny \begin{matrix}
				G \in \mathcal{G}\\
				\mathbb{E}_G (t) =\mu
		\end{matrix}}\quad & \mathbb{E}_{G}  \left[ p(t) (t-R) -q(t)c \right]     \label{eq:robust} \tag{Robust}\\
		\text{ subject to } & (\ref{eq:pfqf})-(\ref{eq:ic}). \nonumber
	\end{align}
	We find an optimal robust mechanism indirectly through a solution to problem (\ref{eq:pw}). The mechanism is presented in Proposition \ref{prop:robustresult} below, followed by discussions on the principal's payoff; a more general statement and the related discussions are delegated to Section \ref{sec:multiplePW}.

	\begin{prop}
		\label{prop:robustresult}
		The following mechanism is optimal within robust mechanisms:
		\begin{enumerate}[leftmargin = 0.5\parindent]
			
			\item If $ \mu \ge R$, then the agent receives the good without being checked.
			
			\item If $ \mu < R$, then proceed as follows:
			\begin{enumerate}
				\item If $t - c < R$, then the principal retains the good.
				
				\item If $t - c \ge R$, then the agent is checked (truth-telling on equilibrium path) and receives the good.
			\end{enumerate}
		\end{enumerate}
	\end{prop}

	We now compare the optimal robust mechanism and the benchmark optimal mechanism in Section \ref{sec:benchmark}, under distribution $F$. First of all, the benchmark optimal mechanism uses the detail-dependent threshold $t^* - c$, whereas the robustly optimal mechanism uses the information-independent threshold $\mu$. In terms of payoffs, the principal's payoff under the benchmark optimal mechanism is 
	given in (\ref{eq:y}), with $s$ replaced by $t$ and $G$ replaced by $F$, i.e.,
	\begin{equation*}
		Y=
		\begin{cases}
			\mu - R & \text{ if } t^* - c \ge R \\
			\int_{R + c }^{\bar{t}} (t - c - R) \mathrm{d} F(t) & \text{ if } t^* - c < R ,
		\end{cases}
		\tag{Y}
	\end{equation*}
	whereas her payoff under the optimal robust mechanism is 
	\begin{equation}
		\label{eq:yrobust}
		\tag{Y-robust}
		Y^{rbst} := 
		\begin{cases}
			\mu - R & \text{ if } \mu \ge R \\
			\int_{R + c }^{\bar{t}} (t - c - R) \mathrm{d} F(t) & \text{ if } \mu < R.
		\end{cases}
	\end{equation}
	Since any robust mechanism is by default feasible, of course we have $Y \ge Y^{rbst}$. Yet we are more interested in how different they are.

	The two payoffs differ only when $R$ satisfies $t^* - c < R \le \mu$\footnote{\ From equation (\ref{eq:t}), we know that 
		$\mathbb{E}(t) = \mathbb{E} (\max \{t - c,t^* - c\}),$
		which implies
		$
		t^* - c \le \mu.
		$}, in which case 
	\begin{equation}
		\label{eq:diff}
		\int_{R + c }^{\bar{t}} (t - c - R) \mathrm{d} F(t) > \mu - R.
	\end{equation}
	If we examine the gross payoffs rather than the net ones, then (\ref{eq:diff}) becomes
	$$
	\int_{\ubar{t}}^{R + c} R \mathrm{d} F(t) + \int_{R + c }^{\bar{t}} (t - c ) \mathrm{d} F(t) > \mu,
	$$
	which is equivalent to saying that
	$
	\mathbb{E}(\max\{R, t-c\}) > \mathbb{E}(t).
	$
	In other words, when $t^* - c < R \le \mu$, the benchmark optimal mechanism sacrifices a cost $c$ when $t > R + c$ to guarantee a payoff $R$ when $t < R + c$. The benefit from the latter is larger than the cost due to the former.
	

	\subsection{Principal-optimal information design}

	We now turn back to information design and characterize the solution to the principal-optimal information design problem below:
	\begin{align}
		\max_{G \in \mathcal{G}_F}\quad \max_{p,q} \quad & \mathbb{E}_{G}  \left[ p(s) (s-R) -q(s)c \right]      \label{eq:po} \tag{PO}\\
		\text{ subject to } & (\ref{eq:pfqf})-(\ref{eq:ic}). \nonumber
	\end{align}
	
	
	\begin{prop}
		\label{prop:principal-optimal}
		Given $R$, $F$ and $c$:
		\begin{enumerate}[leftmargin = 0.5\parindent]
			\item If $t^* - c \geq R$, the principal's payoff is independent of information design. In particular, full information $F \in \mathcal{G}_{F}$ is principal-optimal.

			\item If $t^* - c <R$, an information $G \in \mathcal{G}_{F}$ is principal-optimal if and only if $G^{-}(R + c) = F(R + c)$ and $\mathbb{E}_{G}(s \big\vert [R + c, \bar{t}]) = \mathbb{E}_{F}(s \big\vert [R + c, \bar{t}]) $. In particular, full information $F \in \mathcal{G}_{F}$ is principal-optimal.
			
		\end{enumerate}
	\end{prop}

	Using full information, we know that the principal's payoff under principal-optimal information is given by (\ref{eq:y}), with $s$ replaced by $t$ and $G$ replaced by $F$. The comparison between the principal-optimal payoff and the principal-worst payoff is the same as that between (\ref{eq:y}) and (\ref{eq:yrobust}).\footnote{The agent's payoff does not vary within principal-optimal information. To see it, we still consider two cases. If $t^* - c \ge R$, then the agent receives the good with probability one, whatever the principal-optimal information is. If $t^* - c < R$, then according to the proof of Proposition \ref{prop:principal-optimal}, we know that any principal-optimal information $G$ must maintain $s^* - c < R$. Thus, under the optimal mechanism in Proposition \ref{prop:benchmark}, the agent receives the good if and only if $s - c \ge R$, with probability $1 - G^-(R + c)$. But part 2 of Proposition \ref{prop:principal-optimal} says that this probability is always $1 - F(R+c)$ and is independent of $G$.}

	If we think of the principal and the information designer in our model as the same person, then the spirit of Proposition \ref{prop:principal-optimal} echoes the key message of \cite{kattwinkel2019costless}. Namely, Proposition \ref{prop:principal-optimal} says that the principal who knows $F$ does not profit from any form of information design, e.g., choosing a signal distribution $G$ that is a coarsening of $F$. Somewhat similarly, in a principal-agent model where the principal privately observes a signal that is correlated with the agent's type, \cite{kattwinkel2019costless} show that the principal does not profit from any form of information design, e.g., releasing parts of her information to the agent to manipulate his beliefs. However, the difference is drastic. Information design in \cite{kattwinkel2019costless} pertains to the disclosure of the principal's private signal to the agent who knows his own type, whereas in our setting, the information designer designs the entire uncertainty in the underlying allocation problem, perceived by both the principal and the agent.

	\section{Multiple agents}
	\label{sec:multiple}

	When there are multiple agents, we denote the set of agents by $\mathcal{I} = \{1, \dots, I\}$. Let $\hat{\mathcal{I}}:= \{0, 1, \dots, I\}$, where $0$ stands for the principal. A generic agent or the principal is denoted by $i$; and we abuse the terminology ``agent'' to always say agent $i$. The distribution of agent $i$'s type is $F_i$ over $T_i = [\ubar{t}_i, \bar{t}_i]$. Assume for all $i$ that $0 < \ubar{t}_i < R < \bar{t}_i < \infty$, $R + c \le \bar{t}_i$ and $f_i(t_i) >0$ for all $t_i$. In the special case of $i = 0$, $t_i \equiv R$. Define profiles $\bt := (t_1, \dots, t_I)$ and define $\bt_{-i}$, $\mathbf{T}$ and $\mathbf{T}_{-i}$ in the usual manner.

	Let $\mathcal{G}_{F_i}$ be the set of marginal distributions of unbiased signals for agent $i$. The cost to verify agent $i$'s private signal is $c_i$; in the special case of $i = 0$, $c_i = 0$. Each agent $i$, including the principal, has a critical type $s_i^*$ that is defined in the manner of (\ref{eq:t}) using distribution $G_i \in \mathcal{G}_{F_i}$. 
	 
	 The information designer now has the flexibility to simultaneously choose $G_{i} \in \mathcal{G}_{F_i}$ for every $i\in \mathcal{I}$.  Assume that the information designer's choices of $G_{i}$'s are independent for now, which we will revisit later. Define $\mathbf{G} := (G_1, \dots, G_I)$ and define $\mathbf{G}_{-i}$, $\mathbf{s}$ and $\mathbf{s}_{-i}$ accordingly. 
	
	Given $\mathbf{G}$, the principal's mechanism design problem is as follows:
	\begin{align}
		\max_{\{p_i, q_i\}_{i\in \mathcal{I}}} \quad & \mathbb{E}_{\mathbf{G}} \left\lbrace  \sum_{i = 1}^{I} \left[ p_i(\mathbf{s}) (s_i - R) -q_i(\mathbf{s}) c_i \right] \right\rbrace \label{eq:mdmultiple} \tag{MD-M} \\
		\text{ subject to } \quad
		&0\le q_i(\mathbf{s})\leq p_i(\mathbf{s}) \le 1, \quad \forall \mathbf{s}\in \mathbf{T},  \quad \forall i\in \mathcal{I}, \label{eq:feasibilitymultiple}\\
		&\sum_{i = 1}^I p_i(\mathbf{s}) \le 1, \quad \forall \mathbf{s}\in \mathbf{T}, \\
		&\mathbb{E}_{\mathbf{G}_{-i}} \left[ p_i(s_i, \mathbf{s}_{-i}) \right] \ge \mathbb{E}_{\mathbf{G}_{-i}} \left[ p_i(s_i', \mathbf{s}_{-i}) \right] - \mathbb{E}_{\mathbf{G}_{-i}} \left[ q_i(s_i', \mathbf{s}_{-i}) \right], \nonumber \\
		 & \quad  \quad  \quad  \quad  \quad  \quad  \quad  \quad  \quad  \quad \quad  \quad  \quad  \quad  \quad  \quad \forall s_i, s_i' \in T_i,  \quad \forall i\in \mathcal{I}.  \label{eq:icmultiple}
	\end{align}
	We first recall the solution to this problem due to \cite{ben2014optimal}. As defined in their paper, a \textit{favored-agent mechanism} specifies a \textit{favored agent} $i^*\in \hat{\cI}$ (we treat the principal as an ``agent'' here) and a \textit{threshold} $v^* \in \mathbb{R}$ such that:
	\begin{enumerate}[leftmargin= 0.5\parindent]
		\item If $s_j - c_j< v^*$ for all $j \ne i^*$, then 
		$$
		p_{i^*}(\textbf{s}) = 1, \quad 
		p_j(\textbf{s}) = 0, \forall j\ne i^*; 
		\quad \text{ and }\quad
		q_j(\textbf{s}) = 0, \forall j.
		$$
		
		\item If there exists $j\ne i^*$ such that $s_j - c_j > v^*$ and $s_i - c_i > \max_{k\ne i} s_k - c_k $, then 
		$$
		p_i(\textbf{s}) = q_i(\textbf{s}) = 1, 
		\quad \text{ and } \quad
		p_j (\textbf{s}) = q_j(\textbf{s}) = 0, \forall j\ne i.
		$$
	\end{enumerate}
	
	\begin{prop}[\cite{ben2014optimal}]
		\label{prop:fam}
		The mechanism in each case below is essentially the unique optimal mechanism\footnote{\ We display two cases here to keep consistency with Proposition \ref{prop:benchmark}. If one regard the principal as an ``agent'', then the favored-agent mechanism that favors the ``agent'' who has the largest $s^* - c$ ($R$ for the principal) and adopts the threshold of such $s^* - c$, is essentially the unique optimal mechanism, exactly as in \cite{ben2014optimal}.}:
		\begin{enumerate}[leftmargin = 0.5\parindent]
			\item If $\max_{i \in \cI} \left\lbrace s_i^* - c_i \right\rbrace \ge R$, then implement the favored-agent mechanism with a favored agent $i^* \in \arg\max_{i \in \cI} \left\lbrace s_i^* - c_i \right\rbrace$ and threshold $v^* = s_{i^*}^* - c_{i^*}$.
			
			\item If $\max_{i \in \cI} \left\lbrace s_i^* - c_i \right\rbrace < R$, then implement the favored-agent mechanism with the principal being favored and threshold $v^* = R$.
		\end{enumerate}
	\end{prop}

	\subsection{Agent-optimal information design}
	\label{sec:multipleAO}
	Now we revisit the agent-optimal information design problem in a multi-agent setting. In this setting,  the information designer wants to maximize the agents' \textit{aggregate} payoff, i.e. the total allocation probability.\footnote{\ Such an aggregate agent-optimal information is also a \textit{Pareto} agent-optimal information in the sense that no other information design can make some agent strictly better off without making any agent worse off. We provide a counterexample in Appendix \ref{sec:appPareto} saying that the converse is not true, i.e. a Pareto agent-optimal information needs not be aggregate agent-optimal. If we consider the special case where the prior distributions are identical, i.e. $F_i=F$, and require that the information designer can only choose identical signal distributions, i.e. $G_i=G\in\mathcal{G}_F$,  then obviously Pareto agent optimality is equivalent to aggregate agent optimality.
		\par
		We also explore the situation that the agents have the power to acquire the signals by themselves in Appendix \ref{sec:appendixAO}. We actually study a strategically information acquisition game among agents where those agents themselves are ``information designers''.}
The idea to solve this problem is much similar to the single-agent case characterized in Proposition \ref{prop:agent-optimal}.  Particularly, if there exists some agent $i$ who can be a favored agent ($s_i^*-c\geq R$), then in any case, the principal will allocate this good out. If there does not exist such a favored agent, then for all agents, the information designer will put as much as possible mass on $R+c$ to maximize the total allocation probability. We summarize this in Proposition \ref{prop:aggregate-agent-optimal}.
	\begin{prop}\label{prop:aggregate-agent-optimal}
		Given $R$, $\{F_i\}_{i\in\mathcal{I}}$ and $c$:
		\begin{enumerate}[leftmargin = 0.5\parindent]
			\item If there exists some $i\in\mathcal{I}$, such that $\mathbb{E}_{F_i}(t_i)\geq R$, then for agent $i$, there exists a distribution $G_i\in\mathcal{G}_{F_i}$ such that $s_i^*-c\geq R$. Any profile of signal distributions that has such a component
			is (aggregate) agent-optimal.
			\item If $\mathbb{E}_{F_i} (t_i)<R$ for all $ i\in\mathcal{I}$,  in which case we let $s_i^{\dagger}$ solve  
			\begin{equation*}
				\label{eq:s1}
				R + c = \mathbb{E}_{F_i} \left[t_i \bigg \vert t_i \in \left[s_i^{\dagger}, \bar{t}_i\right] \right]
			\end{equation*}
			for each $i$, 
			then a  profile of signal distributions 
			$\{G_j\}_{j\in\mathcal{I}}$ 
			is agent-optimal if and only if for each $i$,  $G_i$ has an atom $R + c$ with probability $1- F_i(s_i^{\dagger})$. In particular, the following profile of signal distributions is aggregate agent-optimal:
			\begin{equation*}
				\forall j\in \cI, \quad \hat{G}_j(s_j) =\begin{cases}
					F_j(s_j) & \text{ for all } s_j\in[\underline{t}_j,s_j^{\dagger})\\
					F_j(s_j^{\dagger}) & \text{ for all }s_j \in[s_j^{\dagger},R+c)\\
					1&  \text{ for all }s_j \in[R+c,\bar{t}_j].
				\end{cases} 
			\end{equation*}
		\end{enumerate}
	\end{prop}
	We then use an example below to show that \textit{aggregate} agent-optimal information design may affect \textit{individual} agents' payoffs in diverse ways: under some aggregate agent-optimal information, all agents can be equally better off; while under some other aggregate agent-optimal information, one agent gets better off at the cost of the other's payoff.
	
	\begin{example}
		\label{ex:diverse}
		\emph{Suppose there are two agents who have the same underlying type distribution $F$ which is uniform on $[0,1]$ and a common checking cost $c = 0.08$. Suppose the principal's reservation value is $R = 0.4$. In this case, it is straightforward to verify that $t^* = 0.4$ so that $t^* - c < R$. Therefore, under the benchmark optimal mechanism, the principal retains the good with probability $0.48^2$ and assigns the good to the two agents with  probability $ \left(1- 0.48^2\right) \approx 0.76$ in total, and $0.38$ for each.}
		
\paragraph{Both agents are better off.}  
\emph{
	Consider the designed signal distribution profile such that $\hat{G}_1 = \delta\left( \mathbb{E}_F (t) \right) = \delta\left(0.5\right)$ and 
\begin{equation*}
	\hat{G}_2(s) =\begin{cases}
		s & \text{ for all } s\in[0, 0.5)\\
		0.5 & \text{ for all }s \in[0.5, 0.55)\\
		0.6 &  \text{ for all }s \in[0.55, 0.6)\\
		s &  \text{ for all }s \in[0.6,1].
	\end{cases}
\end{equation*}
Namely, $\hat{G}_2$ is derived from $F$ by pooling the probability mass over $(0.5, 0.6)$ at its mean $0.55$.
Then $s_1^* = 0.58$ and $s_2^* = 0.4$. Since $s_1^* - c > R > s_2^* - c$, the optimal mechanism favors agent $1$ and the designed information is aggregate agent-optimal. Under this favored-agent mechanism, agent $1$ receives the good with probability $0.6$ and agent $2$ receives the good with probability $0.4$.
}

\paragraph{One is better off and the other is worse off.} 
\emph{
	Consider another designed signal distribution profile that  $\tilde{G}_1 = \delta\left( \mathbb{E}_F(t) \right) = \delta\left(0.5\right)$ and 
\begin{equation*}
	\tilde{G}_2(s) =\begin{cases}
		s & \text{ for all } s\in[0, 0.3)\\
		0.3 & \text{ for all }s \in[0.3, 0.55)\\
		0.8 &  \text{ for all }s \in[0.55, 0.8)\\
		s &  \text{ for all }s \in[0.8,1].
	\end{cases}
\end{equation*}
Namely, $\tilde{G}_2$ is derived from $F$ by pooling the probability mass over $(0.3, 0.8)$ at its mean $0.55$.
Then $s_1^* = 0.58$ and $s_2^* \approx 0.42$. Since $s_1^* - c > R > s_2^* - c$, the optimal mechanism favors agent $1$ and the designed information is aggregate agent-optimal. Under this favored-agent mechanism, agent $1$ receives the good with probability $0.8$ and agent $2$ receives the good with probability $0.2$.}
	\end{example}

Example \ref{ex:diverse} also serves as a counterexample, showing that the counterpart of Proposition \ref{prop:agentprincipal} does not hold any more. More specifically, $\left(\hat{G}_1, \hat{G}_2\right)$ and $\left(\tilde{G}_1, \tilde{G}_2\right)$ in Example \ref{ex:diverse} are both aggregate agent-optimal. The former leads to a principal payoff of 
$$0.6\cdot 0.5 + \int_{0.6}^1 (s_2 - 0.08) \mathrm{d} s_2 = 0.588,$$
whereas the latter leads to a principal payoff of 
$$0.8\cdot 0.5 + \int_{0.8}^1 (s_2 - 0.08) \mathrm{d} s_2 = 0.564 .$$
Since the two principal payoffs are different, at least one of them is not principal-worst, invalidating the counterpart of Proposition \ref{prop:agentprincipal}. Nevertheless, some agent-optimal information is indeed principal-worst. For example, we will see shortly in the next subsection that $(G_1, G_2) = \left(\delta(0.5), \delta(0.5)\right)$, which is also aggregate agent-optimal, is a principal-worst information and the corresponding principal payoff is $0.5$. 

	\subsection{Principal-related information design}
	\label{sec:multiplePW}
	
	We consider the principal-worst information design problem first. To come straight to the point, we formulate the information design problem and a robust mechanism design problem together: Let $\mathcal{G}_i$ be the set of all cumulative distribution functions defined over $[\ubar{t}_i, \bar{t}_i]$ (then $\mathcal{G}_{F_i} \subseteq \mathcal{G}_i$) and let $\mu_i = \mathbb{E}_{F_i}(t_i)$ be the mean type for agent $i$, which is the only information that the principal can rely on in designing robust mechanisms. The two problems are as follows:
		\begin{align}
		\min_{G_i \in \mathcal{G}_{F_i}, \forall i\in \mathcal{I}}\quad \quad \max_{\{p_i,q_i\}_{i\in \mathcal{I}}} \quad \quad & \mathbb{E}_{\mathbf{G}} \left\lbrace  \sum_{i = 1}^{I} \left[ p_i(\mathbf{s}) (s_i - R) -q_i(\mathbf{s}) c_i \right] \right\rbrace     \label{eq:pwm} \tag{PW-M}\\
		\text{ subject to } & (\ref{eq:feasibilitymultiple})-(\ref{eq:icmultiple})\nonumber\\
	\max_{\{p_i, q_i\}_{i\in \mathcal{I}}} \quad 
	\min_{\tiny \begin{matrix}
			G_i \in \mathcal{G}_{i}, \forall i\in \mathcal{I}\\
			\mathbb{E}_{G_i} (s_i) =\mu_i, \forall i\in \mathcal{I}
	\end{matrix}}\quad & \mathbb{E}_{\mathbf{G}} \left\lbrace  \sum_{i = 1}^{I} \left[ p_i(\mathbf{s}) (s_i - R) -q_i(\mathbf{s}) c_i \right] \right\rbrace   \label{eq:robustm} \tag{Robust-M}\\
	\text{ subject to } & (\ref{eq:feasibilitymultiple})-(\ref{eq:icmultiple}). \nonumber
	\end{align}

	Obviously, the principal-worst information design problem (\ref{eq:pwm}) is a min-max problem, which has a natural connection with (\ref{eq:robustm}) via the max-min inequality. To wit, since (omitting the identical objective function and constraints)
	\begin{eqnarray*}
		& & \quad \max_{\{p_i, q_i\}_{i\in \mathcal{I}}} \quad 
		\min_{\tiny \begin{matrix}
				G_i \in \mathcal{G}_{i}, \forall i\in \mathcal{I}\\
				\mathbb{E}_{G_i} (s_i) =\mu_i, \forall i\in \mathcal{I}
		\end{matrix}} \\
		& \le &  \min_{\tiny \begin{matrix}
				G_i \in \mathcal{G}_{i}, \forall i\in \mathcal{I}\\
				\mathbb{E}_{G_i} (s_i) =\mu_i, \forall i\in \mathcal{I}
		\end{matrix}}
		\quad \quad \max_{\{p_i,q_i\}_{i\in \mathcal{I}}} \quad \text{ (Max-min inequality)}\\
		& \le &  \quad \min_{G_i \in \mathcal{G}_{F_i}, \forall i\in \mathcal{I}}\quad \quad 
		\max_{\{p_i,q_i\}_{i\in \mathcal{I}}} \quad (\mathcal{G}_{F_i} \subseteq \left\lbrace G_i \in \mathcal{G}_i: \mathbb{E}_{G_i} (s_i) =\mu_i \right\rbrace),
	\end{eqnarray*}
	the optimized value of problem (\ref{eq:pwm}) provides an upper bound for the  optimized value of problem (\ref{eq:robustm}). As a result, if a robust mechanism can achieve this upper bound, then it has to be an optimal robust mechanism.

	Our formal result deals with ambiguity sets that are more general than mean-constrained $\mathcal{G}_i$'s. More precisely, for each $i$, we consider an arbitrary ambiguity set $\tilde{\mathcal{G}}_i$ of signal distributions on $[\ubar{t}_i, \bar{t}_i]$ such that
	\begin{enumerate}[leftmargin=0.7\parindent]
		\item [(1)] $\tilde{\mathcal{G}}_i$ contains the null information, i.e., $\delta\left(\mu_i\right)  \in \tilde{\mathcal{G}}_i$, and
		\item [(2)] distributions in $\tilde{\mathcal{G}}_i$ have the same mean, i.e., $\mathbb{E}_{G_i}(s_i) = \mu_i$ for all $G_i \in \tilde{\mathcal{G}}_i$.
	\end{enumerate}
	The two statements of the following proposition are analogous to Propositions \ref{prop:principal-worst} and Proposition \ref{prop:robustresult}, respectively.
	
	\begin{prop}
		\label{prop:robustmultiple}
		The following two statements are true:
		\begin{enumerate}[leftmargin = 0.5\parindent]
			
			\item The principal's payoff under the principal-worst information is 
			\begin{equation}
				\label{eq:robustpayoff} 
				Y^{M} := \max\left\lbrace R, \mathbb{E}_{F_1}(t_1),\dots, \mathbb{E}_{F_I}(t_{I}) \right\rbrace. \tag{Y-M}
			\end{equation}
			Particularly, null information for every agent is principal-worst.
			
			\item For any profile of ambiguity sets $\left\lbrace \tilde{\mathcal{G}}_i \right\rbrace_{i\in \mathcal{I}}$, where each $ \tilde{\mathcal{G}}_i$ satisfies (1) and (2), the robust mechanism that allocates the good to the agent (possibly the principal) who has the highest expected value is optimal within robust mechanisms.
		\end{enumerate}
	\end{prop}

	Proposition \ref{prop:robustmultiple} is even more general than it appears in at least three aspects. First, in the model of \cite{mylovanov2017optimal} which assumes costless ex post verification and limited punishment on lying agents or the model of \cite{li2021mechanism} which accommodates both costly verification and limited punishment, the mechanism in part 2 of Proposition \ref{prop:robustmultiple} is still an optimal robust mechanism\textemdash when the incentive issue is circumvented by the robust mechanism, how severe the punishment is or whether verification is costly becomes irrelevant.

	Second, the problem of allocating multiple homogeneous goods has a similar formulation and a similar solution as the single-good allocation problem; see \cite{ben2019mechanisms} and \cite{chua2019optimal}. Suppose there are $n$ goods to be allocated among $I$ agents. It is straightforward to see that the idea of Proposition \ref{prop:robustmultiple} applies to this multiple-good setting, which leads to an optimal robust mechanism that allocates the $n$ goods to the $n$ agents (including the principal) who have the $n$ highest expected type. 
	
	Third, one may consider the possible correlation among $G_{i}$'s. Roughly speaking, allowing for correlated information design does not affect Proposition \ref{prop:robustmultiple}. Note that, since independent information design is always feasible, allowing for correlated design weakly reduces the principal's payoff under the principal-worst information. Thus, we have a weakly smaller upper bound for the payoffs that can be achieved by robust mechanisms. However, the upper bound (\ref{eq:robustpayoff}) is indeed achievable by the robust mechanism in Proposition \ref{prop:robustmultiple}, which in turn implies that the principal-worst payoff cannot be strictly reduced. Therefore, both the principal-worst payoff and the payoff under the optimal robust mechanism are unchanged, i.e., Proposition \ref{prop:robustmultiple} still holds.\footnote{\ See \cite{he2022correlation} for a study of correlation robust mechanism design in the auction setting. Although the model differs drastically from ours, their idea of using max-min inequality is similar.}


	Finally, the following proposition extends Proposition \ref{prop:principal-optimal}.

	\begin{prop}
		\label{prop:principal-optimal-multiple}
		
		Full information is principal-optimal.
	\end{prop}

	\section{Concluding remarks}

	\label{sec:conclusion}

	In this paper, we examine the information design problems in the allocation setting of \cite{ben2014optimal}, as well as their implications. Many of our insights in the single-agent case extend to the multiple-agent case. 
	
	There are two directions that deserve further investigation. The first one is to study the information acquisition game that we briefly discussed in Section \ref{sec:multipleAO},  though not in the scope of information design as modeled in the current paper. Second, from one agent to multiple agents, when studying principal-optimal information design, we assume independent design across agents. The discussion of possible correlation shall rely on the form of the principal's payoff under the optimal mechanism, which itself, to the best of our knowledge, is an open question.

	%
	%
	%
	
	\begin{spacing}{1.5}
		\bibliographystyle{aea}
		\bibliography{allocation}

@article{bayrak2017optimal,
	title={Optimal allocation with costly inspection and discrete types under ambiguity},
	author={Bayrak, Halil I and G{\"u}ler, Kemal and P{\i}nar, Mustafa {\c{C}}},
	journal={Optimization Methods and Software},
	volume={32},
	number={4},
	pages={699--718},
	year={2017},
	publisher={Taylor \& Francis}
}

@article{bayrak2022optimal,
	title={Distributionally robust optimal allocation with costly verification},
	author={Bayrak, Halil I and Ko{\c{c}}yi{\u{g}}it, {\c{C}}a{\u{g}}ıl and Kuhn, Daniel and P{\i}nar, Mustafa {\c{C}}},
	journal={Working Paper},
	year={2022},
}

@article{ben2014optimal,
  title={Optimal allocation with costly verification},
  author={Ben-Porath, Elchanan and Dekel, Eddie and Lipman, Barton L},
  journal={American Economic Review},
  volume={104},
  number={12},
  pages={3779--3813},
  year={2014},
  publisher={American Economic Association}
}

@article{ben2019mechanisms,
	title={Mechanisms with evidence: Commitment and robustness},
	author={Ben-Porath, Elchanan and Dekel, Eddie and Lipman, Barton L},
	journal={Econometrica},
	volume={87},
	number={2},
	pages={529--566},
	year={2019},
	publisher={Wiley Online Library}
}

@article{bergemann2016informationally,
	title={Informationally robust optimal auction design},
	author={Bergemann, Dirk and Brooks, Benjamin A and Morris, Stephen},
	year={2016},
	publisher={Cowles Foundation Discussion Paper}
}

@article{bergemann2017first,
	title={First-price auctions with general information structures: Implications for bidding and revenue},
	author={Bergemann, Dirk and Brooks, Benjamin and Morris, Stephen},
	journal={Econometrica},
	volume={85},
	number={1},
	pages={107--143},
	year={2017},
	publisher={Wiley Online Library}
}

@article{bergemann2019information,
	title={Information design: A unified perspective},
	author={Bergemann, Dirk and Morris, Stephen},
	journal={Journal of Economic Literature},
	volume={57},
	number={1},
	pages={44--95},
	year={2019}
}

@article{blackwell1953equivalent,
	title={Equivalent comparisons of experiments},
	author={Blackwell, David},
	journal={The annals of mathematical statistics},
	pages={265--272},
	year={1953},
	publisher={JSTOR}
}

@article{brooks2021optimal,
	title={Optimal auction design with common values: An informationally robust approach},
	author={Brooks, Benjamin and Du, Songzi},
	journal={Econometrica},
	volume={89},
	number={3},
	pages={1313--1360},
	year={2021},
	publisher={Wiley Online Library}
}

@article{chen2020information,
	title={Information Design in Optimal Auctions},
	author={Chen, Yi-Chun and Yang, Xiangqian},
	journal={Available at SSRN 3673680},
	year={2020}
}

@article{chua2019optimal,
	title={Optimal Multi-unit Allocation with Costly Verification},
	author={Chua, Geoffrey A and Hu, Gaoji and Liu, Fang},
	journal={Available at SSRN 3407031},
	year={2019}
}

@article{du2018robust,
	title={Robust mechanisms under common valuation},
	author={Du, Songzi},
	journal={Econometrica},
	volume={86},
	number={5},
	pages={1569--1588},
	year={2018},
	publisher={Wiley Online Library}
}

@article{dworczak2019simple,
	title={The simple economics of optimal persuasion},
	author={Dworczak, Piotr and Martini, Giorgio},
	journal={Journal of Political Economy},
	volume={127},
	number={5},
	pages={1993--2048},
	year={2019},
	publisher={The University of Chicago Press Chicago, IL}
}

@article{erlanson2020costly,
	title={Costly verification in collective decisions},
	author={Erlanson, Albin and Kleiner, Andreas},
	journal={Theoretical Economics},
	volume={15},
	number={3},
	pages={923--954},
	year={2020},
	publisher={Wiley Online Library}
}

@inproceedings{epitropou2019optimal,
	title={Optimal On-Line Allocation Rules with Verification},
	author={Epitropou, Markos and Vohra, Rakesh},
	booktitle={International Symposium on Algorithmic Game Theory},
	pages={3--17},
	year={2019},
	organization={Springer}
}

@article{gale1985incentive,
	title={Incentive-compatible debt contracts: The one-period problem},
	author={Gale, Douglas and Hellwig, Martin},
	journal={The Review of Economic Studies},
	volume={52},
	number={4},
	pages={647--663},
	year={1985},
	publisher={Wiley-Blackwell}
}

@article{halac2020commitment,
	title={Commitment versus flexibility with costly verification},
	author={Halac, Marina and Yared, Pierre},
	journal={Journal of Political Economy},
	volume={128},
	number={12},
	pages={4523--4573},
	year={2020},
	publisher={The University of Chicago Press Chicago, IL}
}

@article{he2022correlation,
	title={Correlation-robust auction design},
	author={He, Wei and Li, Jiangtao},
	journal={Journal of Economic Theory},
	volume={200},
	pages={105403},
	year={2022},
	publisher={Elsevier}
}

@article{kattwinkel2019costless,
	title={Costless Information and Costly Verification: A Case for Transparency},
	author={Kattwinkel, Deniz and Knoepfle, Jan},
	journal={Journal of Political Economy},
	year={forthcoming}
}

@article{kleiner2021extreme,
	title={Extreme points and majorization: Economic applications},
	author={Kleiner, Andreas and Moldovanu, Benny and Strack, Philipp},
	journal={Econometrica},
	volume={89},
	number={4},
	pages={1557--1593},
	year={2021},
	publisher={Wiley Online Library}
}

@article{koccyiugit2020distributionally,
	title={Distributionally robust mechanism design},
	author={Ko{\c{c}}yi{\u{g}}it, {\c{C}}a{\u{g}}{\i}l and Iyengar, Garud and Kuhn, Daniel and Wiesemann, Wolfram},
	journal={Management Science},
	volume={66},
	number={1},
	pages={159--189},
	year={2020},
	publisher={INFORMS}
}

@article{li2021mechanism,
	title={Mechanism design with financially constrained agents and costly verification},
	author={Li, Yunan},
	journal={Theoretical Economics},
	volume={16},
	number={3},
	pages={1139--1194},
	year={2021},
	publisher={Wiley Online Library}
}

@article{mookherjee1989optimal,
	title={Optimal auditing, insurance, and redistribution},
	author={Mookherjee, Dilip and Png, Ivan},
	journal={The Quarterly Journal of Economics},
	volume={104},
	number={2},
	pages={399--415},
	year={1989},
	publisher={JSTOR}
}

@article{mylovanov2017optimal,
title={Optimal allocation with ex post verification and limited penalties},
author={Mylovanov, Tymofiy and Zapechelnyuk, Andriy},
journal={American Economic Review},
volume={107},
number={9},
pages={2666--94},
year={2017}
}

@article{roesler2017buyer,
	title={Buyer-optimal learning and monopoly pricing},
	author={Roesler, Anne-Katrin and Szentes, Bal{\'a}zs},
	journal={American Economic Review},
	volume={107},
	number={7},
	pages={2072--80},
	year={2017}
}

@article{townsend1979optimal,
	title={Optimal contracts and competitive markets with costly state verification},
	author={Townsend, Robert M},
	journal={Journal of Economic Theory},
	volume={21},
	number={2},
	pages={265--293},
	year={1979},
	publisher={Elsevier}
}

@article{yang2019buyer,
	title={Buyer-optimal information with nonlinear technology},
	author={Yang, Kai Hao},
	journal={Available at SSRN 3306455},
	year={2019}
}

@article{yang2021efficient,
	title={Efficient demands in a multi-product monopoly},
	author={Yang, Kai Hao},
	journal={Journal of Economic Theory},
	volume={197},
	pages={105330},
	year={2021},
	publisher={Elsevier}
}

@article{vohra2012optimization,
	title={Optimization and mechanism design},
	author={Vohra, Rakesh V},
	journal={Mathematical Programming},
	volume={134},
	number={1},
	pages={283--303},
	year={2012},
	publisher={Springer}
}
	\end{spacing}

	\newpage
	
	\begin{appendices}
		
		\section{Proofs}
		\label{sec:proofs}

		
		\begin{proof}[Proof of Proposition \ref{prop:benchmark}]
			There are three steps. The first step simplifies the problem. The second step restricts our attention to a special class of mechanisms called threshold mechanisms. The third step finds the optimal mechanisms within the class of threshold mechanisms.
			
			\textbf{Step 1. Getting rid of $q(s)$ to simplify the problem.}
			
			From the incentive compatibility constraint (\ref{eq:ic}), we know that
			$$
			\inf_{s\in T} p(s) \ge p(s') - q(s')\quad \text{ for all } s, s'\in T.
			$$
			Let $\varphi := \inf_{s\in T} p(s)$. Then we have
			\begin{equation}
				\label{eq:phige}
				\varphi \ge p(s') - q(s')\quad \text{ for all } s'\in T.
			\end{equation}
			Since checking is costly, (\ref{eq:phige}) must hold with equality for all $s'\in T$. Therefore,
			\begin{equation}
				\label{eq:phieq}
				\varphi = p(s) - q(s)\quad \text{ for all } s\in T.
			\end{equation}
			Plug (\ref{eq:phieq}) into the objective function (\ref{md}) to eliminate $q(s)$. Then we have a simplified problem:
			\begin{align}
				\max_{p} \quad & \mathbb{E}_{G}  \left[ p(s) (s - R) - (p(s) - \varphi) c \right] =  \mathbb{E}_{G}  \left[ p(s) (s - c - R)  \right] + \varphi c \label{eq:sobjective}\\
				\text{ subject to }\quad
				& p(s)\in [0,1], \quad \forall s\in T, \label{eq:spf}\\
				&\varphi = \inf_{s\in T} p(s), \quad \forall s\in T.\label{eq:sic}
			\end{align}

			Problem \eqref{eq:sobjective}-\eqref{eq:sic} is equivalent to the following problem, where we refer to the inner maximization as the \textit{relaxed problem}:
			\begin{eqnarray}
				\max_{\varphi} \quad \max_{p} \quad & \mathbb{E}_{G}  \left[ p(s) (s - c - R)  \right] & + \varphi c \nonumber \\
				\text{ subject to }\quad
				& p(s)\in [0,1], \quad \forall s\in T, & \nonumber \\
				&p(s) \ge \varphi, \quad \forall s\in T.& \label{eq:ric}
			\end{eqnarray}
			In what follows, we study the solution to the relaxed problem first and then the original one. Due to the simplification above, we also call $p(\cdot)$ a mechanism.

			\textbf{Step 2. The solution to the relaxed problem is a \textit{threshold mechanism}.}
			
			\begin{definition}
				\label{def:tm}
				A mechanism $p$ is a \textbf{threshold mechanism} with parameter $\varphi$ if there exists a threshold $v^*$ such that the following two requirements hold for almost all $s\in T$:
				\begin{enumerate}[leftmargin = 1.5 \parindent]
					\item If $s -c < v^*$, then $p(s) = \varphi$.
					\item If $s -c \ge v^*$, then 
					$$
					p(s) = \begin{cases}
						1 	&	\text{ if } s -c \ge R,\\
						0	&	\text{ if } s -c < R.
					\end{cases}
					$$
				\end{enumerate}
			\end{definition}
			
			When $\varphi = 1$, the only feasible mechanism is $p(s) = 1$ for all $s\in T$, which has to be the solution of the relaxed problem. It is straightforward to see that this constant mechanism is a threshold mechanism with, say, $v^* = \bar{t} - c$.
			
			Suppose $\varphi < 1$. Let $p$ be an optimal mechanism in the sense that it solves the relaxed problem. We define a candidate threshold $v^*$:
			$$
			v^* := \inf \{s -c : \quad  p(s') = 1 \text{ for almost all } s' \ge s\}
			$$
			and proceed to argue that $p$ is a threshold mechanism with threshold $v^*$.
			
			We claim that $v^* \ge R$, so that $p$ satisfies the second requirement of threshold mechanism. If $v^* < R$, then for $s\in T$ such that $v^* \le s - c < R$, which has a strictly positive measure, we have $s - c - R < 0$. By the definition of $v^*$, $p(s) = 1$ for almost all $s\in T$ such that $v^* \le s - c < R$. Since the objective function is decreasing in $p$ when $s - c - R < 0$, the principal's payoff improves if $p$ is reduced. More precisely, we consider the following modification of $p$:
			$$
			p'(s) = \begin{cases}
				\varphi &	\text{ if } v^* \le s - c < R, \\
				p(s) 	&	\text{ otherwise.}
			\end{cases}
			$$
			Obviously, $p'$ is feasible. Since $\varphi < 1$, $p'$ delivers a strictly higher payoff to the principal than $p$, which is a contradiction to the optimality of $p$. Therefore, $v^* \ge R$.
			
			To see that $p$ also satisfies the first requirement of threshold mechanism, i.e. $p(s) = \varphi$ for almost all $s - c < v^*$, we consider two cases: (a) $v^* = R$ and (b) $v^* > R$.

			(a) If $v^* = R$, then $s - c < v^*$ means $s - c < R$. Thus, any mechanism with $\int_{\{s - c < v^* : p(s) > \varphi \}} \mathrm{d} G(s) > 0$ can be strictly improved by reducing $p(s)$ to $\varphi$ on the set $\{s - c \le v^* : p(s) > \varphi \}$, contradicting the optimality of $p$. Thus, in this case,  $p(s) = \varphi$ for almost all $s - c < v^*$. 
			
			\medskip
			
			(b) Now suppose $v^* > R$. And suppose to the contrary that there is a positive measure set $D \subseteq [\ubar{t},v^* + c)$ (equivalently, $s -c < v^*$ for all $s\in D$) such that $\varphi < p(s)$ for all $s\in D$. Without loss of generality, there exists an $\epsilon > 0$ such that $\varphi + \epsilon < p(s)$ for all $s\in D$. By the definition of $v^*$, for an arbitrarily small $\delta > 0$, particularly $\delta < v^* - R$, there exists a positive measure set $E \subseteq (v^*+c - \delta, v^* + c)$ such that $p(s) < 1$ for all $s\in E$. Without loss of generality, there exists a small $\gamma > 0$ such that $p(s) < 1 - \gamma $ for all $s\in E$. Without loss of generality, we assume that $D \cap E = \emptyset$ and $\max D < \min E$.\footnote{\ We can make this true by choosing a small enough $\delta$ and re-defining $D$ to exclude $E$.} Consider two numbers $\eta > 0$ and $\xi >0$ such that $\eta \mu(D) = \xi \mu(E)$, and consider the following mechanism:
			$$
			p'(s) = \begin{cases}
				p(s) - \eta &\text{ for all } s\in D, \\
				p(s) + \xi	& \text{ for all } s\in E, \\
				p(s)	&	\text{ otherwise.}
			\end{cases}
			$$
			We can choose $\eta $ and $\xi $ to satisfy $\eta,\xi < \min \{\epsilon, \gamma\}$. Then the new mechanism $p'$ is feasible. Since $\max D < \min E$, we have $s' - c - R < s - c - R $ for all $s\in E$ and all $s'\in D$. Therefore, $p'$ improves on $p$, which contradicts the optimality of $p$. Hence, $p(s) = \varphi$ for almost all $s - c < v^*$.

			\textbf{Step 3. The solution to the original problem.}
			
			Based on Step 2, the principal's objective function can be written as follows:
			\begin{align*}
				& \int_{\ubar{t}}^{\bar{t}} p(s) (s - c - R) \mathrm{d} G(s) + \varphi c \\
				= &	\int_{\ubar{t}}^{v^* + c} \varphi (s - c - R) \mathrm{d} G(s) + \int_{v^* + c }^{\bar{t}} (s - c - R) \mathrm{d} G(s) + \varphi c \quad \text{ (Threshold Mechanism) }	\\
				= & \varphi \left[ \int_{\ubar{t}}^{v^* + c}  (s - c - R) \mathrm{d} G(s) + c \right] + \int_{v^* + c }^{\bar{t}} (s - c - R) \mathrm{d} G(s). 
			\end{align*}
			We need to find the optimal $v^*$ and $\varphi$. Obviously, the optimal mechanism must satisfy
			$$
			\begin{cases}
				\varphi = 1 &	\text{ if } \int_{\ubar{t}}^{v^* + c}  (s - c - R) \mathrm{d} G(s) + c \ge 0, \\
				\varphi = 0 &	\text{ if } \int_{\ubar{t}}^{v^* + c}  (s - c - R) \mathrm{d} G(s) + c < 0.
			\end{cases}
			$$
			When $\varphi = 1$, the principal's objective function becomes $\int_{\ubar{t}}^{\bar{t}}(s - R) \mathrm{d} G(s)$ and the optimal $v^*$ does not matter (but there may be some restriction). When $\varphi = 0$, the principal's objective function becomes $\int_{v^* + c }^{\bar{t}} (s - c - R) \mathrm{d} G(s)$ and the optimal $v^* = R$.  
			
			The former mechanism is optimal if and only if 
			\begin{equation}
				\label{eq:opt2}
				\int_{\ubar{t}}^{\bar{t}}(s - R) \mathrm{d} G(s) \ge \int_{R + c }^{\bar{t}} (s - c - R) \mathrm{d} G(s),
			\end{equation}
			which is equivalent to saying that
			$$\mathbb{E}(s) \ge  \mathbb{E} (\max\{t, R + c\}) - c,$$
			which in turn is equivalent to saying that
			$$ s^* - c \ge R.$$
			One can easily verify that $v^* = s^* - c$ is an eligible threshold such that $\int_{\ubar{t}}^{v^* + c}  (s - c - R) \mathrm{d} G(s) + c \ge 0$ holds in this case (by checking (\ref{eq:opt2})).
			
			Similarly, the latter mechanism is optimal if and only if 
			\begin{equation}
				\label{eq:opt1}
				\int_{\ubar{t}}^{\bar{t}}(s - R) \mathrm{d} G(s) < \int_{R + c }^{\bar{t}} (s - c - R) \mathrm{d} G(s),
			\end{equation}
			which is equivalent to saying that
			$$\mathbb{E}(s) <  \mathbb{E} (\max\{t, R + c\}) - c,$$
			which in turn is equivalent to saying that
			$$ s^* - c < R,$$
			where $s^*$ is defined in (\ref{eq:t}). One can easily verify that $\int_{\ubar{t}}^{R + c}  (s - c - R) \mathrm{d} G(s) + c < 0$ holds in this case (by checking (\ref{eq:opt1})).

			It is straightforward to see that $\varphi = 1$, $v^* = s^* - c$ and $ s^* - c \ge R$ correspond to the first scenario in the statement of the proposition, and $\varphi = 0$, $v^* = R$ and $ s^* - c < R$ to the second. This completes the proof.
		\end{proof}

		
		\begin{proof}[Proof of Lemma \ref{lem:sbiggerthant}]
			Note that
			\begin{eqnarray*}
				\mathbb{E}_G(\max\{s,s^*\}) & = &  \mathbb{E}_G(s) + c \quad \text{ (Definition of $s^*$) }\\
				& = &  \mathbb{E}_F(t) + c \quad \text{ ($G\in \mathcal{G}_F$) }\\
				& = & \mathbb{E}_F(\max\{t,t^*\})  \quad \text{ (Definition of $t^*$) }\\
				& \ge & \mathbb{E}_G(\max\{s,t^*\}). \quad \text{ ($G$ is a MPC of $F$ and $\max\{\cdot,t^*\}$ is convex) }
			\end{eqnarray*}
			Suppose $s^* < t^*$ on the contrary. Since $\mathbb{E}_G(\max\{s, \cdot \})$ is weakly increasing, we have
			\begin{eqnarray*}
				\mathbb{E}_G(\max\{s,s^*\}) < \mathbb{E}_G(\max\{s,t^*\}),
			\end{eqnarray*}
			where the inequality holds strictly because $s^*$ is defined as the greatest threshold that achieves $S := \mathbb{E}_G(\max\{s,s^*\})$. Hence, we have a contradiction.
		\end{proof}


		\begin{proof}[Proof of Lemma \ref{lem:etsmall}]
			Suppose to the contrary that $s^* \ge R+c$. Then we would have
			\begin{eqnarray*}
				\mathbb{E}_{G}(s) & = & \mathbb{E}_{G}(\max\{s,s^*\})-c \quad \text{ (Definition of $s^*$) }\\
				& \geq & \mathbb{E}_{G}(\max\{s,R+c\})-c \quad \text{ ($\mathbb{E}_{G}(\max\{s,\cdot\})$ is weakly increasing) }\\
				& \geq & R+c-c\\
				& = & R.
			\end{eqnarray*}
			Since $G\in\mathcal{G}_F$, we have $\mathbb{E}_{G}(s) = \mu$. Therefore, $\mu \ge R$, a contradiction.
		\end{proof}

		\begin{proof}[Proof of Proposition \ref{prop:agent-optimal}]
			%
			Suppose $\mu\geq R$. Consider the degenerate signal distribution $G=\delta(\mu)$ which assigns probability one to the atom $s = \mu$; $G\in \mathcal{G}_F$. Then,
			$$\mathbb{E}_G (\max\{s, s^*\}) - c=\max\{\mathbb{E}_G (s),s^*\}-c.$$
			By definition of $s^*$, i.e. (\ref{eq:t}), we have
			$\mathbb{E}_G (s) = \max\left\lbrace \mathbb{E}_G (s),s^* \right\rbrace-c,$
			which implies $\mathbb{E}_G (s) = s^* - c$. Since $\mathbb{E}_G(s)=\mu$ for any $G\in \mathcal{G}_F$, we have $\mu = s^* - c$. Therefore, $s^* - c \ge R$ and, consequently, $G$ maximizes the agent's payoff.

			Now suppose $\mu < R$. 
			We prove the ``if'' part first.
			By Lemma \ref{lem:etsmall}, we know that for any $G\in\mathcal{G}_F$, $s^*<R+c$. Then for any $G\in\mathcal{G}_F$, we need to solve (\ref{eq:idao}), i.e. 
			$$
			\max_{G\in \mathcal{G}_F} \quad 1 - G^{-}(R + c).
			$$
			In other words, we need to make $G^{-}(R+c)$ as small as possible.  We claim that for any $G\in\mathcal{G}_F$,  $G^{-}(R+c)\geq F(s^{\dagger})$ or, equivalently, 
			$$
			1 - G^{-}(R+c) \le 1- F(s^{\dagger}).
			$$
			Then if some distribution $G \in \mathcal{G}_F$ could put probability mass $1-F(s^{\dagger})$ at the atom $R + c$, then it must be agent-optimal, which would complete the proof for the ``if'' part.
			
			\begin{cl}
				For any $G\in\mathcal{G}_F$,  $G^{-}(R+c)\geq F(s^{\dagger})$. 
			\end{cl}

			\begin{proof}
				We will use $\hat{G}$ in (\ref{eq:aoi1}) as an intermediate  distribution to obtain the claim. 
				
				Suppose to the contrary that for $G^{-}(R+c) < F(s^{\dagger})$ for some $G\in\mathcal{G}_F$. Since $G$ is a mean-preserving contraction of $F$, we know that for any $s'\in[\underline{t},\bar{t}]$, 
				\begin{eqnarray*}
					\int_{\underline{t}}^{s'}G(s)\mathrm{d}s & \leq & \int_{\underline{t}}^{s'}F(s)\mathrm{d}s  = \int_{\underline{t}}^{s'}\hat{G}(s)\mathrm{d}s.
				\end{eqnarray*}	 
				Particularly, when $s' = s^{\dagger}$, we have
				\begin{eqnarray*}
					\int_{\underline{t}}^{s^{\dagger}}G(s)\mathrm{d}s  & \le & \int_{\underline{t}}^{s^{\dagger}}\hat{G}(s)\mathrm{d}s.
				\end{eqnarray*}
				Since $G^{-}(R + c)< F(s^{\dagger})= \hat{G}(s^{\dagger})$, we know that for every $s\in (s^{\dagger},R+c)$,
				\begin{eqnarray*}
					G(s) & \le & G^{-}(R + c) \quad \text{ ($G$ is weakly increasing) }\\
					& < &  \hat{G}(s^{\dagger}) \\
					& \le & \hat{G}(s). \quad \text{ ($\hat{G}$ is weakly increasing) }
				\end{eqnarray*}
				Therefore,
				\begin{eqnarray*}
					\int_{\underline{t}}^{(R+c)^-}G(s)\mathrm{d}s < \int_{\underline{t}}^{(R+c)^-}\hat{G}(s)\mathrm{d}s.
				\end{eqnarray*}
				Finally, for every $s \in [R+c,\bar{t}]$, we have $G(s) \le 1 = \hat{G}(R+c)$ and thus
				\begin{eqnarray*}
					\int_{\underline{t}}^{\bar{t}}G(s)\mathrm{d}s<\int_{\underline{t}}^{\bar{t}}\hat{G}(s)\mathrm{d}s,
				\end{eqnarray*}
				which contradicts the fact that $G$ and $\hat{G}$ are both in $\mathcal{G}_F$. Hence, $G^{-}(R+c)\geq F(s^{\dagger})$.
			\end{proof}
			Obviously, the particular information $\hat{G} \in \mathcal{G}_F$ in (\ref{eq:aoi1}) puts probability mass $1-F(s^{\dagger})$ at the atom $R + c$, as desired.

			Now we prove the ``only-if'' part for the case of $\mu < R$. Since $\hat{G}$ solves (\ref{eq:idao}) and satisfies $\hat{G}^{-}(R + c) = F(s^{\dagger})$, any solution $G \in \mathcal{G}_F$ to (\ref{eq:idao}) must satisfy $G^{-}(R + c) = F(s^{\dagger})$ as well. We claim that 
			\begin{eqnarray}
				\int_{\ubar{t}}^{(R + c)^{-}} s \mathrm{d} G(s) \ge \int_{\ubar{t}}^{s^{\dagger}} s \mathrm{d} F(s). \label{eq:claim1}
			\end{eqnarray}
			Since $G\in \mathcal{G}_F$, we know that $\mathbb{E}_G(s) = \mathbb{E}_F(s)$. The claim would imply 
			\begin{eqnarray*}
				\int^{\bar{t}}_{R + c} s \mathrm{d} G(s) &\le& \int^{\bar{t}}_{s^{\dagger}} s \mathrm{d} F(s).
			\end{eqnarray*}
			Using $G^{-}(R + c) = F(s^{\dagger})$ again, we have
			\begin{eqnarray*}
				\frac{1}{1- G^{-}(R + c)}\int^{\bar{t}}_{R + c} s \mathrm{d} G(s) & \le & \frac{1}{1 - F(s^{\dagger})}\int^{\bar{t}}_{s^{\dagger}} s \mathrm{d} F(s)\\
				& = & R + c. \quad \text{ (The definition of $s^{\dagger}$) }
			\end{eqnarray*}
			This can occur only if $G$ puts all the probability mass $1- G^{-}(R + c) = 1 - F(s^{\dagger})$ on the atom $R + c$, in which case the inequality holds with equality. This would complete the ``only if'' part.  The rest of the proof verifies the claim (\ref{eq:claim1}):
			\begin{eqnarray*}
				\int_{\ubar{t}}^{(R + c)^{-}} s \mathrm{d} G(s) 
				& = & \left[(R + c) G^{-}(R + c) - 0\right] - \int_{\ubar{t}}^{(R + c)^{-}} G(s) \mathrm{d} s\\
				& = & \left[s^{\dagger} G^{-}(R + c) - 0\right] + \left[(R + c) - s^{\dagger}\right] G^{-}(R + c)  - \int_{\ubar{t}}^{(R + c)^{-}} G(s) \mathrm{d} s\\
				& = & \left[s^{\dagger} G^{-}(R + c) - 0\right] + \int_{(s^{\dagger})^{+}}^{(R + c)^{-}} G^{-}(R + c) \mathrm{d} s - \int_{\ubar{t}}^{(R + c)^{-}} G(s) \mathrm{d} s\\
				& \ge & \left[s^{\dagger} G^{-}(R + c) - 0\right] - \int_{\ubar{t}}^{s^{\dagger}} G(s) \mathrm{d} s \quad \text{ ($G$ is increasing) }\\
				& = & \left[s^{\dagger} F(s^{\dagger}) - 0\right] - \int_{\ubar{t}}^{s^{\dagger}} G(s) \mathrm{d} s \quad (G^{-}(R + c) = F(s^{\dagger}))\\
				& \ge & \left[s^{\dagger} F(s^{\dagger}) - 0\right] - \int_{\ubar{t}}^{s^{\dagger}} F(s) \mathrm{d} s \quad  \text{ ($G$ is a MPC of $F$) }\\
				& = & \int_{\ubar{t}}^{s^{\dagger}} s \mathrm{d} F(s).
			\end{eqnarray*}
			
			\vspace{-5mm}
		\end{proof}


		\begin{proof}[Proof of Proposition \ref{prop:principal-worst}]

			
			Suppose $\mu\geq R$. To examine (\ref{eq:y}), we first note that when $s^* - c = (resp. <) \  R$, we have
			\begin{eqnarray}
				\int_{R + c }^{\bar{s}} (s - c - R) \mathrm{d} G(s) & = (resp. >) & 
				\mu - R. \label{eq:2principal-worst}
			\end{eqnarray}
			Therefore, the principal's payoff is bounded from below by $\mu-R$, which is independent of the information $G$ to be designed.
			
			If an information $G \in \mathcal{G}_{F}$ leads to a threshold $s^*$ such that $s^* \ge R + c$, then the good is allocated to the agent without checking and the principal obtains the lower bound payoff $\mu-R$. Such a $G$ would be principal-worst. So the ``if'' part holds.  
			
			To show the ``only-if'' part, we shall first notice that for the degenerate distribution $G = \delta(\mu)$, we have $s^* = \mu + c \ge R + c$. Therefore, it is a principal-worst information. If any other information $G \in \mathcal{G}_{F}$ is principal-worst, then it has to deliver the lower bound payoff $\mu-R$ to the principal. By (\ref{eq:2principal-worst}), the principal's payoff is strictly higher than $\mu-R$ whenever $s^* - c < R$. Therefore, a principal-worst information must admit a threshold $s^*$ such that $s^* \ge R + c$.
			
			Now suppose $\mu<R$. We have seen in Lemma \ref{lem:etsmall} that for any $G\in\mathcal{G}_F$, we have $s^*<R+c$. Therefore, according to (\ref{eq:y}), the payoff applicable here is $\int_{R + c }^{\bar{s}} (s - c - R) \mathrm{d} G(s)$. It is straightforward to see that
			\begin{equation}
				\label{eq:3principal}
				\int_{R + c }^{\bar{s}} (s - c - R) \mathrm{d} G(s) \ge 0.
			\end{equation}
			Since the degenerate distribution $G = \delta(\mu)$ is in $\mathcal{G}_F$ and $\mu < R + c$, the lower bound payoff of zero is attainable by $G$ and any information which attains the lower bound is principal-worst. Obviously, (\ref{eq:3principal}) holds with equality if and only if $G(R+c) = 1$.
		\end{proof}

		\begin{proof}
			[Proof of Proposition \ref{prop:agentprincipal}]
			The case of $\mu \ge R$ is trivial. When $\mu < R$, we claim that any agent-optimal information $G\in \mathcal{G}_F$ must satisfy
			$
			1 - G(R + c) = 0.
			$
			Then Proposition \ref{prop:principal-worst} would imply that $G$ is also principal-worst.
			
			Suppose to the contrary that $1 - G(R + c)  > 0$. Then 
			\begin{eqnarray*}
				1 - G^{-}(R + c) & = & g(R + c) + \left[ 1- G(R + c) \right]\\
				& = & 1 - F(s^{\dagger}) + \left[ 1- G(R + c) \right] \quad \text{ (Proposition \ref{prop:agent-optimal}) } \\
				& > & 1 - F(s^{\dagger})\\
				& = & 1 - \hat{G}^{-}(R + c). \quad \text{ (\ref{eq:aoi1}) }
			\end{eqnarray*}
			Therefore, the agent's payoff is strictly higher under $G$ than under $\hat{G}$, which is a contradiction to the optimality of $\hat{G}$. This completes the proof.
		\end{proof}

		\begin{proof}[Proof of Proposition \ref{prop:robustresult}]
			It suffices to show that the mechanism in the proposition can (at least) achieve its upper bound (\ref{eq:ypw}), i.e. the principal-worst payoff in Proposition \ref{prop:principal-worst}.

			When $ \mu  \ge R$, then the agent receives the good without being checked and the principal gets $\mu-R$, which is exactly the principal-worst payoff  in the case of $ \mu  \ge R$.
			
			When $\mu < R$, the principal-worst payoff is zero. We proceed to show that the mechanism in the proposition achieves the payoff zero. For any $G \in \mathcal{G}_F$ and $t\in T$, if  $t - c <R$, then the principal retains the good and she gets payoff zero (excluding the reservation value); if $t - c \geq R$,  then the agent is checked and allocated the good and the principal gets $t-c-R\geq0$ which is weakly positive. Thereby, for any $t\in T$, the principal gets (at least) nonnegative payoff. But the principal's payoff is bounded above by zero. Hence, the mechanism in the proposition must achieve the payoff zero. This completes the proof.	
		\end{proof}

		\begin{proof}[Proof of Proposition \ref{prop:principal-optimal}]
			If $t^* - c \ge R$, then we know from Lemma \ref{lem:sbiggerthant} that $s^* - c \ge R$. In this case, the principal's payoff is $\mu - R$, independent of information design.
			
			Now suppose $t^* - c < R$. The principal's payoff is again independent of the specific information $G$ as long as the induced $s^*$ satisfies $s^* - c \ge R$, which is $\mu - R$. It remains to search for the principal-optimal information within 
			$$
			\{G\in \mathcal{G}_{F}: s^* - c < R\},
			$$
			where the principal's payoff is 
			\begin{equation*}
				\int_{R + c }^{\bar{t}} (s - c - R) \mathrm{d} G(s).
			\end{equation*}
			In what follows, we first find out one principal-optimal information, and then the characterization follows immediately. 
			
			Consider the relaxed problem where the constraint $s^* - c < R$ is dropped:
			\begin{equation*}
				\max_{G\in \mathcal{G}_{F}} \quad \int_{R + c }^{\bar{t}} (s - c - R) \mathrm{d} G(s)
			\end{equation*}
			First, for any $G \in \mathcal{G}_{F}$, the support of $G$ is a subset of $[\ubar{t},\bar{t}]$. Without loss of generality, we assume $G$ is defined on $[\ubar{t},\bar{t}]$.
			Note that 
			$$
			\int_{R + c }^{\bar{t}} (s - c - R) \mathrm{d} G(s) = \int_{\ubar{t}}^{\bar{t}} \max \{0, s - c - R\} \mathrm{d} G(s).
			$$
			Since $G$ is a MPC of $F$ and $\max\{0,\cdot\}$ is a convex function, we know that 
			\begin{equation}
				\label{eq:principal-optimal1}
				\int_{\ubar{t}}^{\bar{t}} \max \{0, s - c - R\} \mathrm{d} F(s) \ge \int_{\ubar{t}}^{\bar{t}} \max \{0, s - c - R\} \mathrm{d} G(s).
			\end{equation}
			That is, $F\in \mathcal{G}_{F}$ solves the relaxed problem. Since $t^* - c < R$, the constraint is also satisfied. Hence, $F$ itself is principal-optimal within $\{G\in \mathcal{G}_{F}: s^* - c < R\}$.
			
			Overall, when $t^* - c < R$, the principal's maximal payoff under designed information is 
			$$
			\max\left\lbrace \mu - R, \int_{R + c }^{\bar{t}} (t - c - R) \mathrm{d} F(t)\right\rbrace.
			$$
			By (\ref{eq:payoffcomparison}), we know that the latter is greater, which means that $F$ is principal-optimal when $t^* - c < R$.

			Now we turn to the characterization of the principal-optimal information.
			Since $\max \{0, s - c - R\}$ is piecewise linear in $s$ and has a unique kink at $s = R + c$, any non-trivial mean-preserving contraction of $F$ across $R + c$ makes (\ref{eq:principal-optimal1}) strict. Therefore, for any principal-optimal information $G \in \mathcal{G}_F$, we must have $G^{-}(R + c) = F(R + c)$. In this case, by rewriting the principal's payoff as	
			\begin{equation*}
				\int_{R + c }^{\bar{t}} (s - c - R) \mathrm{d} G(s) 
				= G^{-}(R + c)(R + c) + [ 1 - G^{-}(R + c)] \mathbb{E}_{G}(s \big\vert [R + c, \bar{t}]) - (R + c),
			\end{equation*}
			we know that an information $G \in \mathcal{G}_F$ such that $G^{-}(R + c) = F(R + c)$ is principal-optimal if and only if
			\begin{eqnarray*}
				\mathbb{E}_{G}(s \big\vert [R + c, \bar{t}]) &=& \mathbb{E}_{F}(s \big\vert [R + c, \bar{t}]). 
			\end{eqnarray*}
			This completes the characterization and thus the entire proof.
		\end{proof}


		\begin{proof}
			[Proof of Proposition \ref{prop:robustmultiple}]
			The two statements are proved simultaneously, without figuring out a principal-worst information first. 
			
			First, the principal's payoff under the principal-worst information, denoted by $Y^{PW}$, is weakly worse than the payoff under the ``uniform-null'' information distributions (one feasible choice of information design since every $\tilde{\mathcal{G}}_i$ contains the null information $\delta(\mathbb{E}_{F_i}(t_i))$). That is,
			\begin{equation*}
				\max\left\lbrace R, \mathbb{E}_{F_1}(t_1),\dots, \mathbb{E}_{F_I}(t_{I}) \right\rbrace \ge Y^{PW}.
			\end{equation*}
			
			Second, by the max-min inequality, the principal's worst-case payoff is an upper bound for her payoff from the optimal robust mechanism, whatever it is. Therefore, we have $Y^{PW} \ge Y^{rbst}$, which implies that 
			\begin{equation*}
				\max\left\lbrace R, \mathbb{E}_{F_1}(t_1),\dots, \mathbb{E}_{F_I}(t_{I}) \right\rbrace  \ge Y^{rbst}.
			\end{equation*}
			
			Finally, since the mechanism that allocates the good to the agent (possibly the principal) who has the highest expected value is robust and it can achieve the upper bound payoff (\ref{eq:robustpayoff}), we know that such a mechanism is optimal among robust mechanisms. As a byproduct, (\ref{eq:robustpayoff}) is indeed the principal's payoff under the principal-worst information. 
		\end{proof}

		\begin{proof}
			[Proof of Proposition \ref{prop:principal-optimal-multiple}]
			We need to consider two cases: (1) $\max_{i \in \cI} \left\lbrace t_i^* - c_i \right\rbrace \ge R$ and 
			(2) $\max_{i \in \cI} \left\lbrace t_i^* - c_i \right\rbrace < R$.

			If $\max_{i \in \cI} \left\lbrace t_i^* - c_i \right\rbrace \ge R$, then by Lemma \ref{lem:sbiggerthant}, we must have $\max_{i \in \cI} \left\lbrace s_i^* - c_i \right\rbrace \ge R$ for any profile of designed signal distributions. According to Proposition \ref{prop:fam}, the optimal mechanism admits a favored-agent $i \in \cI$ and has a threshold $v^* = s_i^* - c_i$. In this case, the principal's reservation value is irrelevant in the analysis. The rules of the favored-agent mechanism implies that the principal's payoff in this case is
			\begin{eqnarray}
				&&	\int_{\mathbf{s}: \max_{j\ne i} \left\lbrace s_j - c_j \right\rbrace \le v^*} s_i \mathrm{d} \mathbf{G}(\mathbf{s}) 
				+ \int_{\mathbf{s}: \max_{j\ne i} \left\lbrace s_j - c_j \right\rbrace > v^*} \max_{k \in \mathcal{I}} \left\lbrace s_k - c_k \right\rbrace \mathrm{d} \mathbf{G}(\mathbf{s}) \label{eq:payofffam1}  \\
				& = & \mathbb{E}_{\mathbf{G}} \left( s_i \bigg \vert \max_{j\ne i} \left\lbrace s_j - c_j \right\rbrace \le v^* \right) \prod_{j\ne i} G_{j} \left( v^* + c_j\right)  \nonumber\\
				& & + \ 
				\mathbb{E}_{\mathbf{G}} \left( \max_{k\in \cI} \left\lbrace s_k - c_k \right\rbrace  \bigg \vert \max_{j\ne i} \left\lbrace s_j - c_j \right\rbrace > v^* \right) \left[ 1 - \prod_{j\ne i} G_{j} (v^* + c_j) \right].
				\label{eq:payofffam}
			\end{eqnarray}
			Information design affects (\ref{eq:payofffam}), as a weighted average, in two aspects: the weights and the two weighted terms. We first compare the two weighted terms:

			\begin{cl}
				\label{cl:com}
				$$\mathbb{E}_{\mathbf{G}} \left( s_i \bigg \vert \max_{j\ne i} \left\lbrace s_j - c_j \right\rbrace \le v^* \right)
				\le 
				\mathbb{E}_{\mathbf{G}} \left( \max_{k\in \cI} \left\lbrace s_k - c_k \right\rbrace  \bigg \vert \max_{j\ne i} \left\lbrace s_j - c_j \right\rbrace > v^* \right)
				$$
			\end{cl}
			
			\begin{proof}
				Rewrite the first term of (\ref{eq:payofffam1}) as follows: 
				\begin{eqnarray*}
					&& \int_{\mathbf{s}: \max_{j\ne i} \left\lbrace s_j - c_j \right\rbrace \le v^*} s_i \mathrm{d} \mathbf{G}(\mathbf{s}) \\
					& = & \int_{\mathbf{s}_{-i}: \max_{j\ne i} \left\lbrace s_j - c_j \right\rbrace \le v^*} \int_{s_i} s_i \mathrm{d} G_i(s_i) \mathrm{d} \mathbf{G}_{-i}(\mathbf{s}_{-i})\\
					& = & \int_{\mathbf{s}_{-i}: \max_{j\ne i} \left\lbrace s_j - c_j \right\rbrace \le v^*} \mathbb{E}_{G_{i}} \left(s_i\right) \mathrm{d} \mathbf{G}_{-i}(\mathbf{s}_{-i})\\
					& = & \int_{\mathbf{s}_{-i}: \max_{j\ne i} \left\lbrace s_j - c_j \right\rbrace \le v^*} \mathbb{E}_{G_{i}} \left(\max\left\lbrace s_i - c_i, v^* \right\rbrace \right) \mathrm{d} \mathbf{G}_{-i}(\mathbf{s}_{-i}) \quad \text{ (Definition of $s_i^*$) }.
				\end{eqnarray*}
				Then, we know that
				$$\mathbb{E}_{\mathbf{G}} \left( s_i \bigg \vert \max_{j\ne i} \left\lbrace s_j - c_j \right\rbrace \le v^* \right) 
				=
				\mathbb{E}_{\mathbf{G}} \left( \max\left\lbrace s_i - c_i, v^* \right\rbrace \bigg \vert \max_{j\ne i} \left\lbrace s_j - c_j \right\rbrace \le v^* \right).$$
				Note that for any $s_i$, any $\mathbf{s}' = (s_i, \mathbf{s}_{-i}') \in \left\lbrace \mathbf{s}: \max_{j\ne i} \left\lbrace s_j - c_j \right\rbrace \le v^* \right\rbrace$ and any $\mathbf{s}'' = (s_i, \mathbf{s}_{-i}'') \in \left\lbrace \mathbf{s}: \max_{j\ne i} \left\lbrace s_j - c_j \right\rbrace > v^* \right\rbrace$, we must have
				$$
				\max\left\lbrace s_i - c_i, v^* \right\rbrace \le \max\left\lbrace s_i - c_i, \max_{j\ne i} \left\lbrace s_j'' - c_j\right\rbrace \right\rbrace  = \max_{k\in \cI} \left\lbrace s_k'' - c_k \right\rbrace,
				$$
				which implies that 
				\begin{equation*}
					\mathbb{E}_{\mathbf{G}} \left( \max\left\lbrace s_i - c_i, v^* \right\rbrace  \bigg \vert \max_{j\ne i} \left\lbrace s_j - c_j \right\rbrace \le v^* \right) \\
					\le
					\mathbb{E}_{\mathbf{G}} \left( \max_{k\in \cI} \left\lbrace s_k - c_k \right\rbrace  \bigg \vert \max_{j\ne i} \left\lbrace s_j - c_j \right\rbrace > v^* \right).
				\end{equation*}
				The claim follows immediately.
			\end{proof}

			We have two observations. First, since the weight for the smaller term in (\ref{eq:payofffam}), i.e., $\prod_{j\ne i} G_{j} \left( v^* + c_j\right) $, is increasing in $v^*$/$s_i^*$, Claim \ref{cl:com} implies that the principal prefers a smaller $s_i^*$ to a greater one.

			The second observation is on the two weighted terms in (\ref{eq:payofffam}). The first weighted term is simply $\mathbb{E}_{G_i}(s_i) = \mathbb{E}_{F_i}(t_i)$, which is independent of information design. The second weighted term is the conditional expectation of $\max_{k\in \cI} \left\lbrace s_k - c_k \right\rbrace$, which is obviously a convex function of $\mathbf{s}$. Therefore, the second weighted term is decreasing in each $G_j$ in the second-order stochastic dominance sense, i.e., mean-preserving spread increases its value.
			
			Overall, full information for every $j\in \cI$, i.e., $G_j = F_j$, would (i) minimize the weight for the smaller weighted term in (\ref{eq:payofffam}) by minimizing $s_i^*$ (Lemma \ref{lem:sbiggerthant}) and (ii) maximize the larger weighted term. Hence, such a design is principal-optimal.

			Now we consider the case with $\max_{i \in \cI} \left\lbrace t_i^* - c_i \right\rbrace < R$. If for some agent $i\in \cI$, the designed signal distribution $G_i$ induces an $s_i^*$ such that $s_i^* - c_i \ge R$, then there will be a favored-agent and the principal's payoff takes the form of (\ref{eq:payofffam}). 
			However, since the principal prefers smaller $s_i^*$ to a greater one and prefers mean-preserving spread to mean-preserving contraction, making $s_i^* - c_i \ge R$ (through mean-preserving contraction) is never optimal.
			In contrast, if for all $i\in \cI$, the designed signal distribution induces an $s_i^*$ such that $s_i^* - c_i < R$, then the principal's payoff is
			\begin{equation*}
				R \prod_{i\in \cI} G_i(R + c) + \mathbb{E}_{\mathbf{G}} \left[ \max_{k\in \cI} \left\lbrace s_k - c_k \right\rbrace \bigg \vert \max_{k\in \cI} \left\lbrace s_k - c_k \right\rbrace > R \right] \left[ 1-\prod_{i\in \cI} G_i(R + c)\right].	
			\end{equation*}
			By similar reasons as for (\ref{eq:payofffam}), full information for every agent is principal-optimal. This completes the proof for the second case and thus the entire proposition.
		\end{proof}

		\section{Pareto agent-optimal information}
		\label{sec:appPareto}

		\begin{example}[Pareto agent-optimal information needs not be aggregate agent-optimal]
			
			\emph{
				Consider an example with  $c_i=c$ and for all $i$, $\mu_i<R$.  Consider the following information structure $\{G_i\}_{i\in\mathcal{I}}$ that for some agent $i$, $G_i=\hat{G}_i$ defined in Proposition 8 (Section 6.1) and $G_i=\delta(\mu_j)$ for any $j\neq i$. Under $\{G_i\}_{i\in\mathcal{I}}$, 
				\begin{align*}
					p_i=1-F_i (s_i^{\dagger} ), \text{ and } p_j=0\quad \forall j\neq i.
				\end{align*}
				First, there is no other information structure $\{G_i'\}_{i\in\mathcal{I}}$ such that $p_i'>p_i$, since the good will be allocated to agent $i$ only if $s_i\geq R+c$.  Thus suppose there is a Pareto-improvement information structure  $\{G_i'\}_{i\in\mathcal{I}}$. There must be at least one $j\neq i$ such that $p_j>0$, therefore for agent $j$,
				\begin{align*}
					Prob\{s_j\geq R+c\}>0.
				\end{align*}
				In this case, if we choose the equal tie-breaking rule, the probability of getting the good for agent $i$ is
				\begin{align*}
					p_i'= &\frac{1}{2}(1-F_i (s_i^{\dagger} ))	Prob\{s_j= R+c\}+(1-F_i (s_i^{\dagger} ))	\left(1-	Prob\{s_j\geq R+c\}\right)\\
					=&(1-F_i (s_i^{\dagger} ))\left(\frac{1}{2}Prob\{s_j= R+c\}+1-Prob\{s_j\geq R+c\}\right)\\
					\leq &(1-F_i (s_i^{\dagger} ))\left(+1-\frac{1}{2}Prob\{s_j\geq R+c\}\right)<(1-F_i (s_i^{\dagger} ))=p_i.
				\end{align*}
				Therefore, if we strictly benefit  some agents, we must strictly hurt some other agents. Hence a Pareto improvement is impossible given $\{G_i\}_{i\in\mathcal{I}}$.
			}
		\end{example}
		
		%
		%
	%

			\section{Agents' information acquisition game}
		\label{sec:appendixAO}
		A principal is to decide whether or not to allocate an indivisible good to multiple agents.
		We study the context in which agents
		can   strategically engage in information acquisition. That is, each
	agent $i$ does not know the realization of $t_i$  \emph{ $\acute{a}$ priori}  but can choose to acquire information
		about $t_i$. We assume that agents cannot communicate and can only choose their
		own signal structure.  Formally, each agent $i$ can independently choose a signal  distribution $G_i\in\mathcal{G}_i$. After seeing the agents' information acquisition behavior, the principal  will implement the optimal favored-agent mechanism with respect to the action profile $\mathbf{G}$ (Proposition \ref{prop:fam}).\footnote{\ \cite{yang2019buyer,yang2021efficient} study such a similar information acquisition game among agents in the auction model. }  An agent's utility is simply the probability of receiving the good after the optimal mechanism is executed. 
		
		In what follows, we provide an asymptotically symmetric Nash equilibrium of the game at the information acquisition stage, which describes the strategic behaviors of the agents when acquiring information.
		\begin{prop}
			\label{prop:acquisition}
			Given any $F$, the strategy profile where every agent chooses $G_i=F\in\mathcal{G}_F$, i.e., fully revelation, forms a symmetric Nash equilibrium in the information acquisition game as $n\rightarrow\infty$.
		\end{prop}

		\begin{proof}
			We argue that any deviation $G_i(s_i)$  for agent $i$ would not be profitable.  This argument contains two parts:
			\begin{enumerate}[leftmargin= 0.5\parindent]
				\item First, if agent $i$'s realization  $s_i<\bar{t}$, then agent $i$ has zero probability to win this good. Thus in this case, agent $i$ will gain zero profit. To see this, since $\bar{t}$ is on the support of $F$, there exists a positive $\varepsilon>0$, such that for any other agent $j$,
				\begin{align*}
					\text{Prob}\{s_j\vert s_i<s_j\leq \bar{t} \}>\varepsilon.
				\end{align*}
				Therefore,	the probability of  the largest realization of other agents $\max\{x_{-i}\}$ which is strictly greater than $s_i$ is
				\begin{align*}
					\text{Prob}\{s\vert s=\max\{s_{-i}\}>s_i\}>1-(1-\varepsilon)^n,
				\end{align*}
				which converges to 1 as $n\rightarrow\infty$. 
				\item Second, we consider the case that agent $i$ has the chance to win with  positive probability, that is $s_i=\bar{t}$. However, in this case, we claim that the probability that $s_i=\bar{t}$ under $F$ is no less than that under any other distributions $G_i\in\mathcal{G}_F$. Hence, in this case, any deviation is still not  profitable.
				
				This  directly follows from the definition of the mean-preserving contraction. To see this, since $G_i\in \mathcal{G}_F$, for any $s_i\in[\underline{t},\bar{t}]$, 
				\begin{align*}
					\int_{s_i}^{\bar{t}}F(s)\mathrm{d}s\geq \int_{s_i}^{\bar{t}}G_i(s)\mathrm{d}s
				\end{align*}
				Now suppose that $\text{Prob}_F\{s_i=\bar{t} \}<\text{Prob}_{G_i}\{s_i=\bar{t} \}$. Then as $s_i\rightarrow\bar{t}$,
				\begin{align*}
					\int_{s_i}^{\bar{t}}F(s)\mathrm{d}s<\int_{s_i}^{\bar{t}}G_i(s)\mathrm{d}s,
				\end{align*}
				which contradicts to the definition of the mean-preserving contraction. 
			\end{enumerate}
		\end{proof}

	While for finite agent, there may not exist a pure-strategy Nash equilibrium, which is illustrated by the following   Example \ref{ex:eq}.
		\begin{example}
			[Nonexistence of pure-strategy Nash equilibrium]
			\label{ex:eq}
			
			\emph{
				Let $F$ be an uniform distribution on $[0,1]$, $R=1/4$ and $c=1/8$ with two agents.  
				The maximal $s_{max}^*$ is $\mathbb{E}(t)+c=1/2+1/8=5/8$, which can be induced by any $G\in\mathcal{G}_F$ such that $\text{supp } G\subseteq[0,s^*_{max}]$.  Without loss of generality, let $G$ be
				\begin{align*}
					G(s)=	\begin{cases}
						s, &\text{ if } s\in[0,1/4];\\
						1/4, & \text{ if } s\in[1/4,5/8);\\
						1,& \text{ if } s\in[5/8,1].
					\end{cases}
				\end{align*}
			}

			\emph{Suppose $(G,G)$ is a pure-strategy Nash equilibrium. Any deviation should be not profitable.  Now define $\hat{G}$ to be 
				\begin{align*}
					\hat{G}(s)=	\begin{cases}
						s, &\text{ if } s\in[0,1/4+2\varepsilon];\\
						1/4+2\varepsilon, & \text{ if } s\in[1/4+2\varepsilon,5/8+\varepsilon);\\
						1,& \text{ if } s\in[5/8+\varepsilon,1],
					\end{cases}
				\end{align*}
				for a small enough $\varepsilon>0$. This $\hat{G}$ has two features: first, the correspondent $\hat{s}^*$ is strictly smaller than $s^*_{max}$; second, it puts as much as possible mass above $s^*_{max}$. Each agent has the incentive to deviate from $G$ to $\hat{G}$ since
				\begin{align*}
					1-\left(1/4+2\varepsilon\right)>1/2.
				\end{align*}
			}
			
			\emph{
				Then we show the following claim:}
			\begin{cl}
				Any strategy profile $(G,G)$ with $s^*<s^*_{max}$ can not be an Nash equilibrium. 
			\end{cl}
			\begin{proof}
				Since $s^*<s^*_{max}$ , there must be $s^*$ is strictly within the $\text{supp } G$. Now define $\tilde{G}$ to be
				\begin{align*}
					\tilde{G}(s)=\begin{cases}
						G(s)&\text{ if } s\in[0,s^*-\varepsilon_1];\\
						G(s^*-\varepsilon_1), & \text{ if } s\in[s^*-\varepsilon_1,s^*);\\
						G(s^*+\varepsilon_2)& \text{ if } s\in[s^*,s^*+\varepsilon_2);\\
						G(s),& \text{ if } s\in[s^*+\varepsilon_2,1],
					\end{cases}
				\end{align*}
				where $\varepsilon_1>0$ and $\varepsilon_2>0$ are chosen such that $\tilde{G}\in\mathcal{G}_F$ and $G(s^*)+\varepsilon>G(s^*+\varepsilon_2)>G(s^*)$ for a small enough $\varepsilon>0$. The last condition ensures that $\tilde{s}^*>s^*$.  In this case, each agent has incentives to deviate from $G$ to $\tilde{G}$ since the benefit is much more than the cost.
				\begin{align*}
					G^2(s^*-\varepsilon_1) >\varepsilon.
				\end{align*}
				
			\end{proof}
			
		\end{example}
		\begin{rem*}
			\emph{Even under the case $c>1/4$ that the deviation from $G$ to $\hat{G}$  is non profitable,  we are still hard to say that $(G,G)$ is a pure-strategy Nash equilibrium, since the deviation within the class with $s^*_{max}$ may still be profitable. We leave the characterization of the symmetric pure-strategy (or even mixed-strategy) Nash equilibrium for the future research. }
		\end{rem*}

	\end{appendices}


\end{document}